\documentclass[11pt,a4paper,twoside,reqno]{amsart}
\makeatletter
\def\@settitle{\begin{center}%
		\baselineskip14\p@\relax
		\normalfont\LARGE\scshape\bfseries
		\@title
	\end{center}%
}

\def\@setauthors{%
  \begingroup
  \def\thanks{\protect\thanks@warning}%
  \trivlist
  \centering\footnotesize \@topsep30\p@\relax
  \advance\@topsep by -\baselineskip
  \item\relax
  \author@andify\authors
  \def\\{\protect\linebreak}%
  \authors%
  \ifx\@empty\contribs
  \else
    ,\penalty-3 \space \@setcontribs
    \@closetoccontribs
  \fi
  \endtrivlist
  \endgroup
}

\makeatother
\makeatletter

\def\subsection{\@startsection{subsection}{2}%
	\z@{.5\linespacing\@plus.7\linespacing}{.5\linespacing}%
	{\normalfont\large\bfseries}}

\def\subsubsection{\@startsection{subsubsection}{3}%
	\z@{.5\linespacing\@plus.7\linespacing}{.5\linespacing}%
	{\normalfont\itshape}}
\usepackage[usenames, dvipsnames]{color}
\definecolor{darkblue}{rgb}{0.0, 0.0, 0.45}

\usepackage[colorlinks	= true,
raiselinks	= true,
linkcolor	= darkblue, %MidnightBlue,
citecolor	= Mahogany,
urlcolor	= ForestGreen,
pdfauthor	= {Peyman Mohajerin Esfahani},
pdftitle	= {},
pdfkeywords	= {},
pdfsubject	= {},
plainpages	= false]{hyperref}

\usepackage{dsfont,amssymb,amsmath,subfigure, graphicx,enumitem,multirow} %,etaremune,savesym}
\usepackage{amsfonts,dsfont,mathtools, mathrsfs,amsthm} 
\usepackage{algorithmicx}
\usepackage{algpseudocode}

\allowdisplaybreaks
\date{\today}

\newtheorem{Thm}{Theorem}[section]
\newtheorem{Prop}[Thm]{Proposition}

\newtheorem{Lem}[Thm]{Lemma}

\newtheorem{As}[Thm]{Assumption}

\newtheorem{Rem}[Thm]{Remark}

\usepackage{epsfig} % for postscript graphics files
\usepackage{graphicx}
\usepackage{grffile}
\usepackage{amsmath}
\usepackage{enumitem}
\usepackage{booktabs}
\usepackage{multirow}
\usepackage{array}
\usepackage{cite}
\usepackage{bm}
\usepackage{pdfpages} 
\usepackage{subfigure}
\usepackage{booktabs}
\usepackage{xcolor}
\usepackage{algorithm}
\usepackage{algpseudocode}
\usepackage{wrapfig}
\usepackage{caption}
\usepackage{appendix}
\usepackage{siunitx}

\newcommand{\BDriT}{\tilde{\boldsymbol{r}}_{i,T}}
\newcommand{\BDfiT}{\boldsymbol{f}_{I_i,T-1}}
\newcommand{\BDViT}{\boldsymbol{\mathcal{V}}_{i,T-1}}
\newcommand{\BDWiT}{\boldsymbol{\mathcal{\varpi}}_{i,T-1}}

\begin{document}

\title[]{Distributed Multiple Fault Detection and Estimation in DC Microgrids with Unknown Power Loads}
\author{Jingwei Dong, Mahdieh S. Sadabadi, Per Mattsson,  Andr\'{e} Teixeira}
\thanks{The work was supported by the Swedish Research Council under the grant 2021-06316, the Swedish Foundation for Strategic
Research, and the Knut and Alice Wallenberg Foundation.}
\thanks{Jingwei Dong (jingwei.dong@it.uu.se), Per Mattsson, and Andr\'{e} Teixeira are with the Division of Systems and Control, Uppsala University, Sweden. 
Mahdieh S. Sadabadi is with the Department of Electrical and Electronic Engineering, The University of Manchester, Manchester, U.K. }
\maketitle

\begin{abstract}
    This paper proposes a distributed diagnosis scheme to detect and estimate actuator and power line faults in DC microgrids (e.g., electric-vehicle charging microgrids) subject to unknown power loads and stochastic noise.
    To address actuator faults, we develop an optimization-based filter design approach within the differential-algebraic equation (DAE) framework, which achieves fault estimation, decoupling from power line faults, and robustness against noise.
    In contrast, the estimation of power line faults poses greater challenges due to the inherent coupling between fault currents and unknown power loads, especially under insufficient system excitation, where their effects become difficult to distinguish from measurements. 
    To the best of our knowledge, this is the first study to address this critical yet underexplored issue. 
    Our solution introduces a novel differentiate-before-estimate strategy. 
    A set of diagnosis rules based on the temporal characteristics (i.e., duration of threshold violation) of a constructed residual is developed to distinguish step load changes from line faults.
    Once a power line fault is detected, a regularized least-squares (LS) method is activated to estimate the fault currents, for which we further derive an upper bound on the estimation error.
    Finally, comprehensive simulations validate the effectiveness of the proposed scheme in terms of estimation accuracy and robustness against disturbances and noise under different fault scenarios.
\end{abstract}

\section{Introduction}
Microgrids have emerged as pivotal components of modernized power systems, which play a crucial role in enhancing grid resilience and promoting the integration of renewable energy sources and some controllable loads such as electric vehicle chargers~\cite{mohamad2019investigation}. 
Among them, DC microgrids have gained increasing attention in recent years due to their distinct advantages over AC microgrids, such as higher efficiency, reduced power losses, and greater power transfer capacity~\cite{salehi2019poverty}.
Despite these benefits, the widespread adoption of DC microgrids remains challenging, particularly due to the lack of mature protection schemes to address potential faults in the systems.
The low impedance of DC microgrids makes them vulnerable to faults, which can trigger sudden surges in fault currents and cause severe damage to critical devices. 
If not detected and mitigated promptly, these faults can lead to irreversible damage to the entire system.
As a result, the development of effective fault diagnosis schemes has become a crucial research priority to ensure the reliability and resilience of DC microgrids.

\subsection{Literature Review}
The existing results on fault diagnosis of DC microgrids typically focus on fault detection and isolation.
The most straightforward way is to set detection thresholds based on characteristics of measured currents and voltages, as presented in~\cite{emhemed2016validation}.
While thresholding-based methods are easy to implement, their diagnosis performance is highly sensitive to the selection of thresholds. 
By exploring additional time-domain information from measured signals, methods based on differential currents~\cite{dhar2017fault,meghwani2016non} and traveling wave analysis~\cite{saleh2017ultra} can achieve more reliable diagnosis results compared to simple \mbox{thresholding-based} approaches. 
For instance, the authors in~\cite{meghwani2016non} analyzed the first and second derivatives of fault currents to detect DC cable ground faults in the presence of load and operating mode changes.
However, such derivative-based approaches are sensitive to noise.
In addition to time-domain information, \mbox{frequency-domain} methods, such as fast Fourier transform~\cite{syafi2018real} and wavelet transform~\cite{yao2013characteristic}, are also effective for fault analysis in microgrid systems.
Yet, their reliance on case-specific fault frequency characteristics limits their generalization capability.
In recent years, the availability of large amounts of historical data has also facilitated the adoption of learning-based approaches in DC microgrid fault diagnosis. 
A key challenge, however, is the limited amount of data in fault scenarios. 
We refer interested readers to~\cite{pan2022learning} and~\cite{wang2023data} for further insights.

Different from signal analysis-based and learning-based approaches, model-based fault diagnosis methods utilize more accurate and comprehensive mathematical models of DC microgrids. 
This allows for improved robustness against disturbances and lower reliance on data.
The basic idea of \mbox{model-based} fault diagnosis methods is to construct residual generators by leveraging model information. 
When fed measured signals, outputs of the residual generators (called residuals) can indicate the occurrence of faults. 
Common choices for model-based residual generators include various types of observers, such as Luenberger observers, unknown input observers, and sliding mode observers~\cite{gao2015survey}. 
In addition, residual generation in the framework of \mbox{differential-algebraic} equations (DAEs)~\cite{nyberg2006residual} has also gained attention recently.

Let us first briefly review some~\mbox{observer-based} fault diagnosis methods developed for DC microgrids. 
In~\cite{yao2021unknown}, an unknown input observer was employed to detect arc faults on power lines. 
To enhance robustness against disturbances while maintaining fault sensitivity, the authors in~\cite{wang2020model} utilized the mixed $\mathcal{H}_{\_}/\mathcal{H}_{\infty}$ index in the design of an optimal fault detection observer for DC microgrids.
Fault isolation can be achieved by building a bank of observers, each of which is designed to be sensitive to specific faults.
For instance, the authors in~\cite{wang2022lpv} constructed fault matrices (or signatures) for different types of faults in DC microgrids, and then designed a bank of unknown input observers based on corresponding sensitivity matrices to distinguish fault types.
In~\cite{vafamand2021fusing}, a bank of unscented Kalman filters was developed to estimate system states under faults, whose outputs were then integrated to locate faulty sensors in DC microgrids.
Regarding DAE-based approaches, they offer several advantages over conventional observer-based methods (e.g., Kalman filters, sliding mode observers, and adaptive observers), including: (i) characterizing all possible residual generators for systems described by DAEs, (ii) deriving residual generators of the lowest possible order~\cite{nyberg2006residual}, (iii) facilitating the incorporation of multiple design objectives and constraints, with the corresponding parameters determined within a well-defined optimization framework~\cite{esfahani2015tractable}.
These benefits have also led to the development of anomaly detection methods for microgrid systems within the DAE framework, such as attack detection in~\cite{pan2021dynamic} and ground fault detection in~\cite{dong2024real}.

Note that the aforementioned methods primarily focus on fault detection and isolation, while research on fault estimation for DC microgrids remains limited. 
Fault estimation involves determining the shape and size of faults, which has more stringent requirements than fault detection and isolation. 
The complexity is further exacerbated in DC microgrids with unknown power loads due to the introduction of nonlinearity~\cite{asadi2020fault}.
To address this issue, the authors in~\cite{asadi2020fault} employed a sliding mode observer to estimate actuator and sensor faults in DC microgrids. 
More recently, an adaptive fault estimation method was developed in~\cite{wan2023decentralized} to estimate the system states and faults simultaneously. 
However, both~\cite{asadi2020fault} and~\cite{wan2023decentralized} rely on power line or load current measurements, which require additional sensors that increase system complexity and maintenance costs.
The authors in~\cite{cecilia2021detection} addressed this limitation by proposing a method based on an augmented state observer and an open-loop line current estimator, but the method is confined to sensor attacks.

Based on the above analysis, the problem of multiple fault detection and estimation that simultaneously considers unmeasurable line currents and unknown power loads in DC microgrids remains insufficiently explored. The gap is reflected in two main aspects. 
First, most existing signal analysis-based~\cite{emhemed2016validation,dhar2017fault,meghwani2016non,saleh2017ultra} and model-based methods~\cite{asadi2020fault,wan2023decentralized} are developed under the assumption that power line currents at one or multiple ends are measurable, which requires additional sensing and communication infrastructure and thus increases deployment cost. 
When considering unmeasurable line currents, power line faults are no longer directly observable but must instead be inferred indirectly from the system dynamics.
Second, research on unknown power loads in DC microgrids has mainly focused on system stability analysis and control design, such as those in~\cite{kwasinski2010dynamic,sadabadi2019scalable}. 
In the context of fault detection and estimation, most existing studies simplify the problem by considering only impedance-type loads or assuming measurable load currents, which substantially reduces the difficulty of analysis.     
We further observe that the coupled effects of fault-induced line currents and unknown power loads pose great challenges in distinguishing between them. 
This difficulty is exacerbated under insufficient system excitation (e.g., in the case of incipient faults), where the slow variation of signals leads to an ill-posed estimation problem.  
To the best of our knowledge, this is the first study to address this challenge in the context of DC microgrids.

\subsection{Main Contributions}
This paper proposes a residual-based multiple fault detection and estimation framework for DC microgrids, which takes into account the effects of unknown power loads and stochastic noise.
We employ the DAE framework for residual generation, given its aforementioned advantages.
Furthermore, considering the scalability issue inherent in centralized diagnosis methods, the proposed diagnosis framework is implemented in a distributed architecture.
Specifically, an independent diagnosis component is designed for each distributed generation (DG) unit within the considered DC microgrid, significantly improving its reliability.
The contributions of the paper are summarized as follows:
\begin{itemize}
    \item \textbf{Multiple fault estimation framework.} 
    We develop a multiple fault estimation framework for DC microgrids, which accounts for both actuator and power line faults in the presence of unknown power loads and stochastic noise.  
    This framework is built on fault estimation filters developed within the DAE framework, where the filter design is formulated as a tractable optimization problem~(Proposition~\ref{prop: fa est}).  

    \item \textbf{Differentiation strategy for power line faults and step power load changes.} 
    Building on the practical assumption of piecewise-constant power loads, we propose a novel strategy to distinguish between power line faults and step load changes by exploiting the temporal characteristic differences of their induced residuals, specifically, the duration of threshold violation (Proposition~\ref{prop: diff_P_fl}). This differentiation process is further formalized into a set of diagnosis rules (Equation~\eqref{eq: disgnosis_rules}), facilitating the subsequent faulty line current estimation.
   
    \item \textbf{Faulty line current estimation with an explicit error bound.} 
    To address the ill-posed issue in faulty line current estimation, we formulate a regularized least-squares (LS) optimization method. 
    An analytical solution is further derived to enable real-time estimation (Proposition~\ref{prop: analytical sol}).
    Furthermore, we establish a computable bound on the expectation of the estimation error, explicitly characterizing its dependence on load and fault variations (Theorem~\ref{Thm}).
\end{itemize}

The structure of the remaining parts of the paper is organized as follows. Section~\ref{sec: 2} introduces the dynamics of DC microgrids and the problem statement. 
Section~\ref{sec: 3} presents the design of the actuator fault estimator, while Section~\ref{sec: 4} develops the faulty line current estimation method. The proposed methods are validated through simulations in Section~\ref{sec: 5}. Finally, Section~\ref{sec: 6} concludes the paper. For readability, some technical proofs are relegated to Appendix.

\textbf{Notation.} The sets~$ {\Bbb N}$,~${\Bbb R}~( {\Bbb R}_+)$, and~${\Bbb R}^n$ denote \mbox{non-negative} integers, (positive) reals, and the space of~$n$ dimensional \mbox{real-valued} vectors, respectively.
The space of~$n \times n$ dimensional symmetric matrices is denoted by~$\mathbb{S}^n$.
The identity matrix of size $n$ is denoted by ${\bf I}_n$.
For a random variable~$\chi$, the probability law and the expectation are denoted by~$\textbf{Pr}[\chi]$, and~$\textbf{E}[\chi]$, respectively.
For a vector~$x = [x_1,\dots,x_{n}]^{\top} \in {\Bbb R}^n$, the $2$-norm of~$x$ is~$\|x\|_2 = \sqrt{\sum^n_{i=1} x^2_i}$. 
For a matrix~$A \in {\Bbb R}^{m \times n}$, its transpose and pseudo-inverse are denoted by $A^{\top}$ and $A^{\dag}$, respectively. 
The $2$-norm of $A$ is denoted by~$\|A\|_2$, which represents the largest singular value of $A$, i.e.~$\bar{\sigma}(A)$. 
If~$A$ is a square matrix, its smallest and largest eigenvalues are denoted by~$\underline{\lambda}_A$ and $\bar{\lambda}_A$, respectively.
We use $Q \succ 0 (Q \prec 0)$ to denote a positive (negative) definite matrix~$Q$.
For a discrete-time signal $r(k)$, the stacked data vector of the length $T$ is denoted by~$\boldsymbol{r}_T(k) = [r^{\top}(k-T+1) ~\dots ~r^{\top}(k)]^{\top}$.

%%%%%%%%%%%%%%%%%%%%%%%%%%%%%%%%%%%%%%%%%%%%%%%%%%%%%%%%%%%%%%%%%%%%%%%%%%%
%%%%%%%%%%%%%%%%%%%%%%%%%%%%%%%%%%%%%%%%%%%%%%%%%%%%%%%%%%%%%%%%%%%%%%%%%%%
%                        Problem formulation
%%%%%%%%%%%%%%%%%%%%%%%%%%%%%%%%%%%%%%%%%%%%%%%%%%%%%%%%%%%%%%%%%%%%%%%%%%%
\section{Model Description and Problem Statement} \label{sec: 2}
We recall the dynamic model of a Kron-reduced DC microgrid composed of $n$ power-electronics-interfaced DG units, which are connected by $m$ power lines. The index sets of DG units and lines are denoted as $\mathbb{N}_{G}$ and $\mathbb{N}_{L}$, respectively. 
Fig.~\ref{fig: DCMG} illustrates the architecture of a DC microgrid with $3$ DG units and $2$ power lines. 
In the subsequent parts of this section, we will elaborate on the dynamics of each individual component within the microgrid.

\subsection{DG Dynamics}
Each DG unit includes a DC input voltage source $V_{dc,i}$, a DC-DC buck converter, and a resistive-inductive-capacitive filter with the parameters ($R_{t,i}, L_{t,i}, C_{t,i}$).
This configuration is illustrated in the dashed box of Fig.~\ref{fig: DCMG}, which depicts the circuit diagram of DG unit $1$ connected to DG unit $2$ via power line $1$. 
Using the circuit theory, for each $i \in \mathbb{N}_G$, the dynamics of DG unit $i$ are presented by \cite{nahata2020passivity}: 
\begin{equation}
	\left\{
    \begin{array}{l}
		C_{t,i}\dot{V_i}(t)=I_{t,i}(t)-P_{i}(t)/V_i(t)-\sum_{k=1}^m \mathbb{B}_{ik} I_k(t) ,\\
		L_{t,i}\dot{I}_{t,i}(t)=-V_i(t)-R_{t,i} I_{t,i}(t)+u_i(t)+f_{a,i}(t),	
	\end{array}
    \right.	
		%\end{split}
	\label{DG}
\end{equation}
where $V_i(t)\in {\Bbb R}$ and $I_{t,i}(t)\in {\Bbb R}$ are the voltage at the Point of Common Coupling (PCC)~$i$ and the filter current, respectively.
The unknown power demand is denoted by $P_{i}(t)\in {\Bbb R}_+$. 
The control signal~$u_i(t) = \varphi_i(t) V_{dc,i}$, where $\varphi_i(t)$ is the duty cycle of the buck converter. 
The current of power line~$k$ is represented by $I_k(t)\in {\Bbb R}$ where $k \in \mathbb{N}_{L}$. 
The microgrid topology is encoded in the incidence matrix element $\mathbb{B}_{ik}$, which indicates the connection between power line~$k$ and DG unit $i$. 
In particular, $\mathbb{B}_{ik} = 1$ if DG unit $i$ is the positive end of~line $k$, $\mathbb{B}_{ik} = -1$ if DG unit $i$ is the negative end of~line $k$, and $\mathbb{B}_{ik} = 0$ otherwise.
Therefore, the term $\sum_{k=1}^m \mathbb{B}_{ik} I_k(t)$ is the total current injected into DG unit~$i$ from its neighbors and represents the physical couplings of DG unit~$i$ with neighboring DG units.
The additive actuator fault in DG unit $i$ is denoted by $f_{a,i}(t) \in {\Bbb R}$.

\begin{figure}[t]
    \centering
    \includegraphics[width=0.5\linewidth]{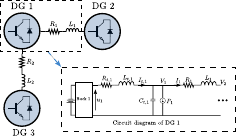}
    \caption{\small Structure of a DC microgrid.}
    \label{fig: DCMG}
\end{figure}

\subsection{Power Line Dynamics}
The DG units are coupled to their neighbors through resistive-inductive power lines. For $k \in \mathbb{N}_L$, the dynamics of the line current~${I}_{k}(t)$ are governed by
\begin{align}\label{line}			
	 \dot{I}_{k}(t)=&-\frac{R_{k}}{L_k} {I}_{k}(t) +\frac{1}{L_k}\sum_{j=1}^n \mathbb{B}_{jk}V_j(t) + f_{L,k}(t),    	
\end{align}
where $R_k$ and $L_{k}$ denote the line resistance and inductance, respectively.  
The fault on line $k$ is denoted by $f_{L,k}(t) \in {\Bbb R}$.

\begin{As}[Unmeasurable line currents and uncertain parameters]
    The line currents are unmeasurable, and the parameters~$R_k$ and~$L_{k}$ can be identified with uncertainty in the \mbox{fault-free} condition.
\end{As}

According to the superposition principle, the line current~$I_k(t)$ consists of a fault-free component $I_{k,h}(t)$, a fault component $I_{k,f}(t)$, and a noise term~$\epsilon_k$ induced by uncertainty, i.e.,~$I_k(t) = I_{k,h}(t) + I_{k,f}(t) + \epsilon_k(t)$.

\begin{Rem}[Applicability beyond additive faults]
    Although the power line faults considered in~\eqref{line} are additive, other types of faults, such as \mbox{short-circuit} faults that change the dynamics of~$I_{k}$, can be represented within this framework as well. 
    A pole-to-ground \mbox{short-circuit} fault will be demonstrated in Section~\ref{sec: 6} to illustrate this applicability.
\end{Rem}

\subsection{Stabilizing Voltage Controller}
For the purpose of voltage control, each DG unit is equipped with a local voltage controller (VC), which is given by:
\begin{equation}\label{eq: control}
    \begin{split}
    \left\{
    \begin{array}{l}
        u_i(t) = k_{1,i}V_i(t)+k_{2,i} I_{t,i}(t)+k_{3,i} v_i(t),\\
        \dot{v}_i(t) =V_i^*-V_i(t),
    \end{array} \right.
    \end{split}
\end{equation}
where $V_i^*$ is a reference voltage provided by a higher-level controller, $K_i =[k_{1,i} \ k_{2,i} \ k_{3,i}]$ is the voltage control gain vector.
Design methods for~$K_i$ can be found in~\cite{sadabadi2019scalable,tucci2017line}.

\subsection{Closed-loop Dynamics}
With~\eqref{DG}-\eqref{eq: control}, the closed-loop dynamics of DG~$i$ for~$i \in \mathbb{N}_{G}$ can be written as: 
\begin{equation}\label{eq: Faulty DG}
\left\{
    \begin{array}{l}
    \dot{x}_i(t)=A_{i} x_i(t)+B_{i} V^*_i+ D_i d_i(t) + E_i f_{a,i}(t) + \delta_i(t),\\
    y_i(t)=C_i x_i(t) + \zeta_{i}(t),
	\end{array}
\right.
\end{equation}
where $x_i=[V_i \ I_{t,i} \ v_i]^T$, $\delta_i$, $y_i$, and $\zeta_i$ denote the state, process noise, output, and measurement noise, respectively. 
The signal~$d_i$ consists of the power load~$P_{i}$ and currents injected from all power lines connected to DG unit $i$, i.e.,
\begin{align*}
    d_i(t)&= P_{i}(t)/ V_i(t)+\sum_{k=1}^m \mathbb{B}_{ik} I_{k}(t) \\
          &= P_{i}(t)/ V_i(t)+ \sum_{k=1}^m(\mathbb{B}_{ik} I_{k,h}(t)+\epsilon_k(t)) + f_{I,{i}}(t),
\end{align*}
where $f_{I,{i}}(t) = \sum_{k=1}^m \mathbb{B}_{ik} I_{k,f}(t)$ denotes the aggregated faulty line current affecting DG unit $i$.
The system matrices in~\eqref{eq: Faulty DG} are as follows:
\begin{equation*}
\begin{split}\label{ss}
    &A_{i}=\left[\begin{array}{ccc}
    0 &\frac{1}{C_{t,i}} & 0\\ 
    \frac{k_{1,i}-1}{L_{t,i}} & \frac{k_{2,i}-R_{t,i}}{L_{t,i}} & \frac{k_{3,i}}{L_{t,i}}\\
    -1 &0&0\end{array}\right], 
    ~B_{i}=\left[\begin{array}{c} 
    0\\ 0\\1\end{array}\right], 
    ~D_{i}=\left[\begin{array}{c} 
    - \frac{1}{C_{t,i}} \\ 0 \\0\end{array}\right], 
    ~E_i = \begin{bmatrix}
        0 \\ \frac{1}{L_{t,i}} \\ 0
    \end{bmatrix}, 
    ~C_i={\bf I}_3.
\end{split}
\end{equation*}
Since the voltage~$V_i$ and the filter current~$I_{t,i}$ are measurable~\cite{tucci2016decentralized}, and~$v_i$ is the controller variable, we have full state measurement.
Let us further introduce the following assumption on noise.
\begin{As}[Stochastic noise]\label{as: uncertainty_noise}
    The noise terms~$\delta_i$, $\epsilon_k$, and~$\zeta_i$ for $i \in \mathbb{N}_G$ and $k \in \mathbb{N}_L$ are mutually uncorrelated zero-mean stochastic processes of unknown distributions. 
    Their covariance matrices denoted by~$\Sigma_{\delta_i}$, $\Sigma_{\epsilon_k}$ and~$\Sigma_{\zeta_i}$ are known. 
\end{As}

The covariance matrices~$\Sigma_{\delta_i}$, $\Sigma_{\epsilon_k}$ and~$\Sigma_{\zeta_i}$ can be computed from available information about the local models, sensors, and possibly on historical data through Monte Carlo-based approaches~\cite{boem2018plug}, and a discussion of their computational approach can be found in~\cite{robert1999monte}.
In addition, we account for unmodeled system uncertainties within the stochastic process noise, as commonly adopted in literature~\cite[Section 3.4]{ding2008model}.

\subsection{Problem Statement.} 
This work aims to develop a distributed method to detect and estimate both the actuator fault~$f_{a,i}$ and the aggregated faulty line current~$f_{I,i}$ in DG unit $i$, while accounting for the unknown power load~$P_{i}$ and the noise terms~$\delta_i$,~$\epsilon_k$, and~$\zeta_i$. 
Estimating $f_{a,i}$ is relatively straightforward, as it is the only unknown signal in the dynamics of the measurable current $I_{t,i}$ in~\eqref{DG} except for noise.
However, it is challenging to even detect $f_{I,i}$ because the line current~$I_k$ is unavailable and~$f_{I,i}$ is coupled with~$P_{i}$ in~$d_i$, as described in~\eqref{eq: Faulty DG}. 
From an algebraic perspective, $f_{I,i}$ and $P_i$ are fundamentally indistinguishable under the following two scenarios:

\begin{enumerate}
    \item \textbf{Arbitrarily varying~$P_{i}$ and $f_{I,i}$.} If~$P_{i}(t)$ and $f_{I,i}(t)$ are unconstrained, different combinations of~$P_{i}(t)$ and $f_{I,i}(t)$ can yield identical effects. 
    For instance, for two different pairs $(P_{i}(t),f_{I,i}(t))$ and $(P^{\prime}_{i}(t),f^{\prime}_{I,i}(t))$, it holds that 
    \begin{align*}
        \frac{P_{i}(t)}{V_i(t)} + f_{I,i}(t) = \frac{P^{\prime}_{i}(t)}{V_i(t)} + f^{\prime}_{I,i}(t),
    \end{align*}
    if $P^{\prime}_{i}(t) =(P_{i}(t)/V_i(t)+f_{I,i}(t)-f^{\prime}_{I,i}(t))V_i(t)$. 
    \item \textbf{Constant~$V_i$.} If $V_i(t)$ is constant, separation of $f_{I,i}$ and~$P_{i}$ is theoretically impossible because $V_i$ provides no variation that can be exploited to distinguish them from the aggregated signal~$P_{i}/V_i + f_{I,i}$.
    This is close to the well-known persistence of excitation phenomena for LTI systems, as discussed in~\cite[Theorem 3.5]{van2022multiple}.
    It is also reflected in the rank condition following~\eqref{eq: parity-space relation with fault} in Subsection~\ref{subsec: line current est}.
\end{enumerate}

Beyond the above algebraic conditions, the indistinguishability between $P_i$ and $f_{I,i}$ also has a physical origin in DC microgrid dynamics. 
As shown in Fig.~\ref{fig: DCMG}, both load variations and line faults act as current-type disturbances injected at the same point, and therefore share identical spatial signatures. 
Moreover, the voltage regulation loop suppresses voltage variations, resulting in insufficient excitation in the measured signals and further weakening the separability of $P_i$ and $f_{I,i}$.
Fortunately, in practical DC microgrids, $V_i$ does vary due to faults and load changes, and power loads do not evolve arbitrarily but typically follow specific patterns.  
Motivated by this observation, we impose the following assumption on~$P_{i}$ throughout the paper.

\begin{As}[Piece-wise constant power load]\label{as: PWC d}
    The power load $P_{i}$ is a piece-wise constant signal that allows multiple step changes over time, with each value persisting for a sufficiently large duration.
\end{As}

Note that assuming constant or piecewise-constant power loads is a standard modeling practice in both control design and fault diagnosis for DC microgrids (e.g.,~\cite{HZW17,sadabadi2019scalable}). This assumption is also practically motivated, as many real-world loads, such as data centers, electric-vehicle charging stations, motor drives, and inverter-based AC loads~\cite{kwasinski2010dynamic,mohamad2019investigation}, can be reasonably characterized as instantaneous constant power loads.
We emphasize that no assumption is imposed on the fault signals.

\begin{figure}[t]
    \centering
    \includegraphics[width=0.5\linewidth]{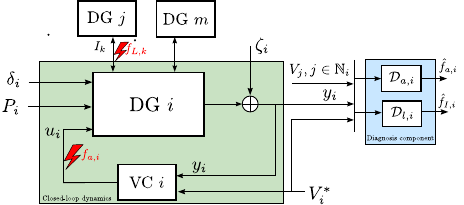}
    \caption{\small Structure of DG $i$ with the diagnosis component: block~$\mathcal{D}_{a,i}$ for the actuator fault~$f_{a,i}$ and block~$\mathcal{D}_{l,i}$ for the aggregate faulty line current~$f_{I,i}$, where $\mathbb{N}_i$ denote the set of neighboring units of DG $i$.}
    \label{fig: Sys}
\end{figure}

We are now in the position to formally present the problems studied in this paper.
To address actuator and power line faults in DC microgrids, we design a diagnosis component for each DG unit, as illustrated in Fig.~\ref{fig: Sys}. 
The diagnosis component consists of two blocks:~$\mathcal{D}_{a,i}$ used to estimate the actuator fault~$f_{a,i}$ and $\mathcal{D}_{l,i}$ used to estimate the aggregate faulty line current~$f_{I,i}$. 
The inputs to the diagnosis component are known signals, including~$V^*_i,~y_i$, and~$V_j$ received from neighboring DG units.
The outputs are the estimates of $f_{a,i}$ and $f_{I,i}$, denoted by~$\hat{f}_{a,i}$ and $\hat{f}_{I,i}$, respectively. 

Our objective is to design $\mathcal{D}_{a,i}$ and $\mathcal{D}_{l,i}$ to achieve \mbox{real-time} estimation of~$f_{a,i}$ and~$f_{I,i}$ in DG unit~$i$. 
Specifically, the design requirements for $\mathcal{D}_{a,i}$ are as follows:
\begin{enumerate}
    \item The effects of influencing factors (e.g., $f_{I,i}$ and~$P_{i}$) on the estimated value~$\hat{f}_{a,i}$ should be decoupled;
    \item The estimation result should be robust to noise signals~$\delta_i$, $\epsilon_k$, and~$\zeta_i$;
    \item Given the stochastic nature of noise, the steady-state estimation error of~$f_{a,i}$ should converge to zero in expectation, i.e., $\lim\limits_{t \rightarrow \infty} | \textbf{E}[\hat{f}_{a,i}(t)-f_{a,i}(t)] | \rightarrow 0$.
\end{enumerate}
The design requirements for $\mathcal{D}_{l,i}$ include:
\begin{enumerate}
    \item The effect of $f_{a,i}$ is decoupled from the estimated value~$\hat{f}_{I,i}$;
    \item The estimation result should be robust to noise signals~$\delta_i$, $\epsilon_k$, and~$\zeta_i$;
    \item Given the possibly ill-posed issue due to the load term~$P_i/V_i$, rather than enforcing convergence of the estimation error to zero, we instead require the estimation error for $f_{I,i}$ to be bounded in expectation, as follows:
    \begin{align*}
    \begin{split}
        \left| \textbf{E}\left[ 
        \hat{f}_{I,i}(t)-\bar{f}_{I,i}(t)
        \right] \right|  
        \leq \mathcal{C}\left(y_i,V^*_i,V_j,P_{i},f_{I,i}\right),
        \end{split}
    \end{align*}
    where $\bar{f}_{I,i}(t)$ represents the average value of $f_{I,i}$ over a period, $\mathcal{C}\left( \cdot \right)$ is a time-varying bound dependent on the system and power line dynamics~\eqref{line}-\eqref{eq: Faulty DG}, input signals to~$\mathcal{D}_{l,i}$, and properties of unknown signals $P_{i}$ and $f_{I,i}$.   
\end{enumerate}

%%%%%%%%%%%%%%%%%%%%%%%%%%%%%%%%%%%%%%%%%%%%%%%%%%%%%%%%%%%%%%%%%%%%%%%%%%
%%%%%%%%%%%%%%%%%%%%%% Section: Main Results %%%%%%%%%%%%%%%%%%%%%%%%%%%%
%%%%%%%%%%%%%%%%%%%%%%%%%%%%%%%%%%%%%%%%%%%%%%%%%%%%%%%%%%%%%%%%%%%%%%%%%%
\section{Estimation of Actuator Faults}\label{sec: 3}
In this section, we provide the design method for $\mathcal{D}_{a,i}$, where a fault estimation filter is developed as the core solution.
To facilitate the filter design, the \mbox{state-space} model~\eqref{eq: Faulty DG} is reformulated into the DAE form:
\begin{align}\label{eq: DAE fa}
    H_i(p) X_i +  \mathcal{B}_i Y_i + \mathcal{E}_i f_{a,i} + \omega_i = 0,
\end{align}
where $p$ is the derivative operator, i.e., $\dot{x}(t) = p x(t)$, the augmented variables~$X_i = [x_i^{\top} ~d_i]^{\top}$,~$\omega_i = [\delta^{\top}_i ~\zeta^{\top}_i]^{\top}$, and $Y_i = [y_i^{\top} ~V^*_i]^{\top}$. 
The polynomial matrix $H_i(p)$ is defined as
\begin{align*}
    H_i(p) = pH_{i,1} + H_{i,0} = 
            \begin{bmatrix} -p{\bf I}_3+A_i &D_i \\ 
                             C_i     &{\bf 0}_{3 \times 1}
            \end{bmatrix},
\end{align*}
where 
\begin{align*}
     H_{i,1} = \begin{bmatrix} 
                -{\bf I}_3     &{\bf 0}_{3 \times 1} \\ 
                {\bf 0}_{3 \times 3} &{\bf 0}_{3 \times 1}
             \end{bmatrix}, \quad 
    H_{i,0} = \begin{bmatrix} 
                A_i &D_i \\ 
                C_i &{\bf 0}_{3 \times 1}
            \end{bmatrix}.
\end{align*}
The matrices $\mathcal{B}_i$ and $\mathcal{E}_i$ are given by 
\begin{align*}
    \mathcal{B}_i = \begin{bmatrix}
        {\bf 0}_{3 \times 3}    &B_i\\
        -{\bf I}_3              &{\bf 0}_{3 \times 1}
    \end{bmatrix}, \quad
    \mathcal{E}_i = \begin{bmatrix}
        E_i\\
        {\bf 0}_{3 \times 1}
    \end{bmatrix}.
\end{align*}

Then, we consider the estimation filter for actuator faults in the following form:
\begin{align}\label{eq: filter}
    \hat{f}_{a,i} = -\frac{N_i(p)\mathcal{B}_i}{a(p)} Y_i,
\end{align}
where the polynomial row vector~$N_i(p) = \sum^{d_{N}}_{j=0} p^j N_{i,j}$, $N_{i,j} \in {\Bbb R}^{1 \times 6}$ and $d_{N}$ is the degree of~$N_i(p)$. 
The denominator~$a(p)$ is a polynomial defined as
$a(p) = p^{d_a} + \sum^{d_a-1}_{j=0} p^j a_j$,
where the coefficient $a_j \in {\Bbb R}$ and~$d_a$ denotes the degree of~$a(p)$. 
For simplicity of filter design, $a(p)$ is given with all roots lying in the left-half plane and is identical for all~$\mathcal{D}_{a,i}$. 
Its degree is set as~$d_a = d_{N}+1$ to ensure that the filter is strictly proper.

By multiplying the left-hand side of~\eqref{eq: DAE fa} by~$N_i(p)/a(p)$, $\hat{f}_{a,i}$ can be expressed as
\begin{align}\label{eq: residual}
\begin{split}
    \hat{f}_{a,i} 
    = \frac{N_i(p)H_i(p)}{a(p)}X_i + \frac{N_i(p)\mathcal{E}_i}{a(p)}f_{a,i} + \frac{N_i(p)}{a(p)}\omega_i.
\end{split}
\end{align}
Note that~\eqref{eq: filter} can be used to generate $\hat{f}_{a,i}$, as all elements are measurable or known, while~\eqref{eq: residual} explicitly characterizes the mapping relations from $X_i$, $f_{a,i}$, and $\omega_i$ to~$\hat{f}_{a,i}$, thus providing the foundation for designing the filter parameters.

Recall the first design requirement for $\mathcal{D}_{a,i}$, which stipulates that $P_{i}$ and $f_{I,i}$ have no influence on~$\hat{f}_{a,i}$.  
To achieve this, we introduce the following condition
\begin{subequations}
\begin{align}
    N_i(p)H_i(p) = 0, \label{eq: decouple con}
\end{align}
which ensures that the signal~$X_i$ implicitly containing $P_{i}$ and~$f_{I,i}$ is completely decoupled from $\hat{f}_{a,i}$. 
To address the second design requirement regarding robustness against stochastic noise, we employ the~$\mathcal{H}_2$ norm approach.
For a linear system driven by white noise with zero mean, the~$\mathcal{H}_2$ norm of its transfer function represents the asymptotic variance of the output~\cite{scherer2002multiobjective}.
Therefore, the second design requirement can be realized by constraining the~$\mathcal{H}_2$ norm of $N_i(p)/a(p)$ as follows
\begin{align}
    \left\Vert \frac{N_i(p)}{a(p)} \right\Vert^2_{\mathcal{H}_2} &\leq \gamma_i, \label{eq: noise con} 
\end{align} 
where~$\gamma_i \in {\Bbb R}_+$ is an upper bound. 
To guarantee convergence of the fault estimate, i.e., $\hat{f}_{a,i} \rightarrow f_{a,i}$ in the steady state, we impose a unity steady state gain constraint on the transfer function~$N_i(p)\mathcal{E}_i/a(p)$, which is
\begin{align} 
    \left. \frac{N_i(p)\mathcal{E}_i}{a(p)} \right\vert_{p=0} = 1. \label{eq: estimation con} 
\end{align}
\end{subequations}

The design requirements for $\mathcal{D}_{a,i}$ have been translated into constraints~\eqref{eq: decouple con}-\eqref{eq: estimation con} on the mapping relations.
In the following proposition, we further formulate a tractable optimization problem for solving the parameters of~$N_i(p)$ based on~\eqref{eq: decouple con}-\eqref{eq: estimation con}. 
Before proceeding, the observable canonical form of $N_i(p)/a(p)$, denoted by~$\{A_r,B_{r,i},C_r\}$, is provided to facilitate computation of the~$\mathcal{H}_2$ norm. Here,~$A_r$ is given by the coefficients of~$a(p)$, which we keep the same for all~$i$, while~$B_{r,i}$ is different for different $i$.
Specifically, the matrices are given by  
\begin{align}\label{eq: SS of filter}
\begin{split}
    &A_r = \begin{bmatrix}
        0 &\dots &0 &-a_0\\
        1 &\dots &0 &-a_1\\
        \vdots &\ddots &\vdots &\vdots\\
        0 &\dots &1 &-a_{d_{N}}
    \end{bmatrix},
    ~B_{r,i}=\begin{bmatrix}
        N_{i,0} \\ N_{i,1} \\ \vdots \\ N_{i,d_{N}}
    \end{bmatrix}, 
    ~C_r = \begin{bmatrix}
        0 &\dots &0 &1
    \end{bmatrix}.
\end{split}
\end{align}
The design approach of the actuator fault estimation filter~\eqref{eq: filter} is then presented in the following proposition.

\begin{Prop}[Actuator fault estimation]\label{prop: fa est}
    Consider the closed-loop dynamics of DG unit $i$ (for $i \in \mathbb{N}_G$) in~\eqref{eq: Faulty DG} subject to the actuator fault $f_{a,i}$.
    The design conditions~\eqref{eq: decouple con}-\eqref{eq: estimation con} for the estimation filter structured in~\eqref{eq: filter} can be equivalently formulated into the following linear programming problem:
\begin{subequations}\label{eq: opt fa}
\begin{align}
        \min ~&\gamma_i \notag \\
        \textup{s.t.} ~&N_{i,j} \in {\Bbb R}^{1 \times 6}, j\in \{0,1,\dots,d_{N}\}, Q_i \in {\Bbb S}^{d_N+1},  
        ~\gamma_i \in {\Bbb R}_+,  \notag \\
        &\begin{bmatrix}
        N_{i,0} &N_{i,1} &\dots &N_{i,d_{N}}
    \end{bmatrix}\bar{H_i}=0, \label{eq: opt fa1}\\
        &N_{i,0}\mathcal{E}_i=a(0),\label{eq: opt fa2}\\
        &\begin{bmatrix}
            A_r Q_i+Q_i A_r^{\top} &B_{r,i} \\
            B^{\top}_{r,i}  &-{\bf I}_6
        \end{bmatrix} \prec 0, ~\begin{bmatrix}
            \gamma_i &C_r Q_i \\ Q_i C_r^{\top} &Q_i 
        \end{bmatrix} \succ 0, \label{eq: opt fa3}
\end{align}
\end{subequations}
where $\bar{H_i}$ is given by
\begin{align*}
    \bar{H_i} = \begin{bmatrix}
        H_{i,0} &H_{i,1} &{\bf 0} &\dots &{\bf 0}\\
        {\bf 0} &H_{i,0} &H_{i,1} &\dots &{\bf 0}\\
        \vdots &\ddots  &\ddots &\ddots &\vdots\\
        {\bf 0} &\dots &{\bf 0} &H_{i,0} &H_{i,1}
    \end{bmatrix}.
\end{align*} 
\end{Prop}

\begin{proof}
    The proof is similar to that of Theorem 3.1 from~\cite{dong2023multimode}, and only a brief sketch is provided. 
    First, based on the polynomial multiplication rule, condition~\eqref{eq: opt fa1} enforces the coefficients of the product $N_i(p)H_i(p)$ to be zero, which guarantees that condition~\eqref{eq: decouple con} is satisfied. 
    Second, condition~\eqref{eq: opt fa2} is obtained by setting the operator $p=0$ in condition~\eqref{eq: estimation con}, yielding the steady-state gain of the corresponding transfer function. 
    Finally, the linear matrix inequalities in~\eqref{eq: opt fa3} follow from classical results on the $\mathcal{H}_2$ norm, which ensures the satisfaction of~\eqref{eq: noise con}. This completes the proof.
\end{proof}

It is worth noting that there always exist feasible filters satisfying~\eqref{eq: decouple con}-\eqref{eq: estimation con}.
First, conditions \eqref{eq: decouple con} and \eqref{eq: estimation con} can be satisfied simultaneously because, from the structures of~$D_i$ and $E_i$, the fault direction does not lie in the disturbance subspace, thereby guaranteeing disturbance decoupling and fault sensitivity~\cite[Theorem 6.2]{ding2008model}.
Second, condition~\eqref{eq: noise con} requires the corresponding transfer function to be stable and strictly proper so that its $\mathcal{H}_2$ norm is finite, which can be ensured by an appropriate choice of $a(p)$.
In addition, the feasibility of problem (10) depends on the rank conditions of $\bar{H}_i$ and $\mathcal{E}_i$.
Specifically,~\eqref{eq: opt fa1} admits nontrivial solutions only if~$\bar{H_i}$ is not of full row rank, which can be ensured by selecting $d_N$ such that $6 (d_N+1) > \text{rank}(\bar{H}_i)$. 
Meanwhile, since all roots of $a(p)$ are in the left-half plane, $a(0) \neq 0$ is guaranteed.
Consequently, condition \eqref{eq: opt fa2} requires that $\mathcal{E}_i$ lies outside the column range space of $\bar{H}_i$, i.e., $\text{rank}([\bar{H_i} ~\mathcal{E}_i]) > \text{rank}(\bar{H_i})$. 
Otherwise, any solution satisfying $N_{i,0}\bar{H}_i = 0$ leads to $N_{i,0}\mathcal{E}_i=0$ as well, and~\eqref{eq: opt fa2} can not be satisfied.

\begin{Rem}[Differences with previous work]
    In our previous work~\cite{dong2024real}, disturbances were non-decouplable, and the filter design was formulated as a quadratic programming problem with a fixed denominator. 
    In contrast, this work (i) exploits structural disturbance decouplability via~\eqref{eq: opt fa1}, (ii) incorporates stochastic noise via the $\mathcal{H}_2$-norm method, and (iii) allows joint numerator--denominator design since the parameters of $a(p)$ are embedded in~$A_r$, leading to a bi-linear formulation when the denominator is also treated as a design variable. 
\end{Rem}

%%%%%%%%%%%%%%%%%%%%%%%%%%%%%%%%%%%%%%%%%%%%%%%%%%%%%%%%%%%%%%%%%%%%%%%%
%%%%%%%%%%%%%%%%%%%%%%%%%%%%%%%%%%%%%%%%%%%%%%%%%%%%%%%%%%%%%%%%%%%%%%%%
\section{Estimation of aggregate faulty line current}\label{sec: 4}
In this section, we present the design method for the diagnosis block $\mathcal{D}_{l,i}$, which estimates the aggregate faulty line current~$f_{I,i}$ induced by faults on power lines connected to DG unit~$i$ under the piecewise-constant power load assumption.
As shown in Fig.~\ref{fig: PF est}, $\mathcal{D}_{l,i}$ consists of four sub-blocks:
\begin{enumerate}
    \item Pre-filter $1$ that decouples the actuator fault $f_{a,i}$ and generates output~$r_{i} \in {\Bbb R}$.
    \item A set of line current estimators to estimate the fault-free part of the total line currents $\sum_{k=1}^m \mathbb{B}_{ik} I_k$ injected into DG unit~$i$, and produces~$\sum_{k=1}^m \mathbb{B}_{ik} \hat{I}_{k,h}$.
    \item Pre-filter~$2$ that produces processed output $\hat{r}_{I,i} \in {\Bbb R}$.
    \item A faulty current estimator that generates the estimated value~$\hat{f}_{I,i}$ using the difference between $r_{i}$ and $\hat{r}_{I,i}$.
\end{enumerate}

In what follows, we provide a detailed elaboration on the functions and design methods for these sub-blocks.

\begin{figure}[t]
    \centering
    \includegraphics[width=0.6\linewidth]{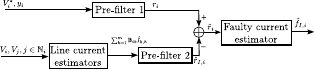}
    \caption{\small Structure of the diagnosis block $\mathcal{D}_{l,i}$.}
    \label{fig: PF est}
\end{figure}

\subsection{Design of Pre-filter $1$}
The design of Pre-filter $1$ is similar to that of the actuator fault estimation filter~\eqref{eq: filter}. 
We first reformulate the \mbox{closed-loop} dynamics~\eqref{eq: Faulty DG} into the DAE form:
\begin{align}\label{eq: DAE fl}
    \mathcal{H}_i(p) \mathcal{X}_i +  \mathcal{B}_i Y_i + \mathcal{G}_i d_i + \omega_i = 0,
\end{align}
where~$\mathcal{X}_i = [x_i^{\top} ~f_{a,i}]^{\top}$.
The polynomial matrix $\mathcal{H}_i(p)$ is obtained by replacing $D_i$ with $E_i$ in $H_i(p)$ from previous DAE~\eqref{eq: DAE fa}. 
The matrix~$\mathcal{G}_i$ is defined as $\mathcal{G}_i = [D^{\top}_i ~{\bf 0}_{3 \times 1}^{\top}]^{\top}$.
The rest terms remain consistent with~\eqref{eq: DAE fa}. 
The structure of \mbox{Pre-filter}~$1$ mirrors~\eqref{eq: filter} and is given by
\begin{align}\label{eq: filter fl}
    r_{i} = -\frac{\mathcal{N}_i(p)\mathcal{B}_i}{a(p)} Y_i,
\end{align}
where $\mathcal{N}_i(p) = \sum^{d_{\mathcal{N}}}_{j=0} p^j \mathcal{N}_{i,j}$ with design parameters~$\mathcal{N}_{i,j} \in {\Bbb R}^{1 \times 6}$ and degree~$d_{\mathcal{N}}$. 
The denominator~$a(p)$ is the same as that in~\eqref{eq: filter}. 
From~\eqref{eq: DAE fl} and~\eqref{eq: filter fl}, $r_{i}$ also equals 
\begin{align*}
    r_{i} = \frac{\mathcal{N}_i(p)\mathcal{H}_i(p)}{a(p)}\mathcal{X}_i + \frac{\mathcal{N}_i(p)\mathcal{G}_i}{a(p)}d_i +\frac{\mathcal{N}_i(p)}{a(p)} \omega_i.
\end{align*}

Recall the first two design requirements for $\mathcal{D}_{l,i}$. 
To decouple~$f_{a,i}$ and suppress the effects of~$\omega_i$, similar conditions as~\eqref{eq: decouple con}-\eqref{eq: estimation con} are employed here to design $\mathcal{N}_i(p)$, which are
\begin{align*}
     &\mathcal{N}_i(p)\mathcal{H}_i(p) = 0,
    ~\left. \frac{\mathcal{N}_i(p)\mathcal{G}_i}{a(p)} \right\vert_{p=0} = 1, ~\left\Vert \frac{\mathcal{N}_i(p)}{a(p)} \right\Vert^2_{\mathcal{H}_2} \leq \tilde{\gamma}_i, 
\end{align*}
where~$\tilde{\gamma}_i \in \mathbb{R}_+$ is an upper bound.
Subsequently, the coefficients of $\mathcal{N}_i(p)$ can be solved using the method in Proposition~\ref{prop: fa est}. 
Since $\mathcal{X}_i$ has been decoupled by the derived Pre-filter $1$, $r_{i}$ becomes
\begin{align}\label{eq: ri}
    r_{i} =  \frac{\mathcal{N}_i(p)\mathcal{G}_i}{a(p)}\left(\sum_{k=1}^m \mathbb{B}_{ik}I_{k} +  \frac{P_{i}}{V_i}\right) +\frac{\mathcal{N}_i(p)}{a(p)} \omega_i,
\end{align}
with $d_i(t)= P_{i}(t)/ V_i(t)+\sum_{k=1}^m \mathbb{B}_{ik} I_{k}(t)$ defined after~\eqref{eq: Faulty DG}.

\subsection{Design of Line Current Estimators}
Based on the power line dynamics given in~\eqref{line}, an \mbox{open-loop} estimator is employed here to estimate the fault-free component of~$I_k$ for each $k \in \mathbb{N}_L$, as follows: 
\begin{equation}\label{eq: Ik_est}	
	 \dot{\hat{I}}_{k,h}(t)=-\frac{R_{k}}{L_k} \hat{I}_{k,h}(t)+ \frac{1}{L_k}\sum_{j=1}^n \mathbb{B}_{jk}V_j(t).   
\end{equation}
In the absence of faults and noise, the dynamics of the estimation error $\tilde{I}_{k} = I_{k} - \hat{I}_{k,h}$ can be described by: $\dot{\tilde{I}}_{k}(t)=-R_k / L_k \tilde{I}_{k}(t)$.
Since $\tilde{I}_{k}$ converges to zero asymptotically, we can replace $I_{k,h}$ with $\hat{I}_{k,h}$ in the subsequent analysis.
When taking faults and noise into account, $\tilde{I}_{k}$ becomes:
\begin{align}\label{eq: Ik}
    \tilde{I}_k = I_k - \hat{I}_{k,h} \approx I_k - I_{k,h}=I_{k,f}+ \epsilon_k.
\end{align}
Given the asymptotic convergence speed, we suppose $\tilde{I}_k = I_{k,f} + \epsilon_k$ for the remainder of the paper.

\subsection{Design of Pre-filter $2$}
Pre-filter $2$ is derived from~\eqref{eq: ri} and is given by:
\begin{align}\label{eq: prefilter 2}
    \hat{r}_{I,i} = \frac{\mathcal{N}_i(p)\mathcal{G}_i}{a(p)} \sum_{k=1}^m \mathbb{B}_{ik} \hat{I}_{k,h},
\end{align}
where the input is an estimate of the aggregate \mbox{fault-free} line currents, the output~$\hat{r}_{I,i}$ represents the estimated contribution of $\sum_{k=1}^m \mathbb{B}_{ik} I_{k,h}$ to~$r_i$.
According to~\eqref{eq: ri},~\eqref{eq: Ik}, and~\eqref{eq: prefilter 2}, the residual~$\tilde{r}_i$ in Fig.~\ref{fig: PF est} encapsulating the information of $f_{I,i}$ and $P_{i}$ is then generated by subtracting~$\hat{r}_{I,i}$ from~$r_i$, yielding 
\begin{align}\label{eq: til_ri}
\begin{split}
    \tilde{r}_{i} &= r_{i} - \frac{\mathcal{N}_i(p)\mathcal{G}_i}{a(p)}\sum_{k=1}^m \mathbb{B}_{ik}\hat{I}_{k,h} \\
    &= \frac{\mathcal{N}_i(p)\mathcal{G}_i}{a(p)}\left(f_{I,i}  +   \frac{P_{i}}{V_i}\right) +\frac{\mathcal{N}_i(p) [\mathcal{G}_i ~{\bf I}_6]}{a(p)}
    \begin{bmatrix}
        \sum_{k=1}^m \mathbb{B}_{ik} \epsilon_k \\ \omega_i
    \end{bmatrix},
\end{split}
\end{align}
where the aggregate faulty line current $f_{I,i} = \sum_{k=1}^m \mathbb{B}_{ik}I_{k,f}$ as previously defined.
To facilitate the subsequent estimation process, we further transform the second line of~\eqref{eq: til_ri} into its corresponding observable canonical state-space form:
\begin{align}\label{eq: ss PF est}
     &\dot{x}_{\tilde{r}_i}(t)  =  A_r x_{\tilde{r}_i}(t) + B_{\mathcal{G},i}\left(f_{I,i}(t)  +\frac{P_{i}(t)}{V_i(t)} \right) + B_{\varpi,i}\varpi_i(t), \notag \\
     &\tilde{r}_{i}(t) = C_r x_{\tilde{r}_i}(t),
\end{align}
where $x_{\tilde{r}_i}(t) \in {\Bbb R}^{d_\mathcal{N}+1}$ and 
$\varpi_i = \begin{bmatrix}
    \sum_{k=1}^m \mathbb{B}_{ik} \epsilon_k &\omega_i^{\top}
\end{bmatrix}^{\top}$. 
Matrices $A_r$ and $C_r$ are specified in~\eqref{eq: SS of filter}, the input matrices~$B_{\mathcal{G},i}$ and $B_{\varpi,i}$ are constructed as: 
\begin{align*}
    &B_{\mathcal{G},i} = 
    \begin{bmatrix}
        \mathcal{N}^{\top}_{i,0} &\mathcal{N}^{\top}_{i,1} &\dots &\mathcal{N}^{\top}_{i,d_{\mathcal{N}}}
    \end{bmatrix}^{\top} \mathcal{G}_i, \\
    &B_{\varpi,i} =
    \begin{bmatrix}
        \mathcal{N}^{\top}_{i,0} &\mathcal{N}^{\top}_{i,1} &\dots &\mathcal{N}^{\top}_{i,d_{\mathcal{N}}}
    \end{bmatrix}^{\top} \begin{bmatrix}
        \mathcal{G}_i &{\bf I}_6
    \end{bmatrix}.
\end{align*}

\subsection{Design of Faulty Line Current Estimator}\label{subsec: line current est}
The three sub-blocks designed above serve as preparatory steps for estimating~$f_{I,i}$.
As illustrated in Fig.~\ref{fig: PF est}, the derived~$\tilde{r}_i$ is processed by the estimator in the final sub-block to compute~$\hat{f}_{I,i}$.
In this final stage, a discrete-time estimation approach is developed, and thus the state-space model~\eqref{eq: ss PF est} is discretized with a sampling period~$t_s$.
For consistency and simplicity, the system matrix notations of the \mbox{discrete-time} model remain unchanged. 
To address the potential ill-posedness issue in estimating $f_{I,i}$ mentioned in the problem statement session, the proposed approach proceeds through two sequential phases: 
\begin{enumerate}
    \item \textbf{Detection and differentiation phase.}  
    Before the occurrence of any power line faults, i.e., $f_{I,i}=0$, a residual~$\tilde{\Upsilon}_{i,T}$ derived from the estimate of $P_{i}$ is introduced to detect both step load changes and power line faults. 
    The discrimination between the two events is further achieved by leveraging their distinct transient behaviors in~$\tilde{\Upsilon}_{i,T}$, in particular, the duration of threshold violation.
    \item \textbf{Estimation phase.} 
    Once a power line fault is detected, a dedicated regularized LS estimator is activated to estimate~$f_{I,i}$. 
\end{enumerate}
The methodologies employed in each phase are elaborated in the remaining parts of this subsection.

\subsubsection{Detection and differentiation of power line faults and step load changes}
The proposed detection and differentiation approach relies on power load estimation before faults occur on the power lines connected to DG unit~$i$.
To this end, suppose that~$f_{I,i} = 0$ and~$P_{i}$ is constant initially based on Assumption~\ref{as: PWC d}. 
To obtain power load estimation, we construct the following parity-space relation over a sliding window $T$ based on~\eqref{eq: ss PF est}: 
\begin{align}\label{eq: PS1}
    \BDriT(k) = &\mathcal{O}_{i,T} x_{\tilde{r}_i}(k-T+1) + \mathcal{Z}_{i_1,T} \BDViT(k-1) P_{i} + \mathcal{Z}_{i_2,T} \BDWiT(k-1) , 
\end{align}
where~$\BDriT(k)$, $\BDViT(k-1)$, and $\BDWiT(k-1)$ are the stacked data vectors of $\tilde{r}_{i}$, $1/V_i$, and $\varpi_i$, respectively.
Toeplitz matrices $\mathcal{Z}_{i_1,T}$, $\mathcal{Z}_{i_2,T}$, and the $T$-order observability matrix $\mathcal{O}_{i,T}$ are defined as:
\begin{align*}
    &\mathcal{Z}_{i_1,T} = \begin{bmatrix}
        0 &0 &\dots &0\\
        C_rB_{\mathcal{G},i} &0 &\dots &0\\
        C_rA_rB_{\mathcal{G},i} &C_rB_{\mathcal{G},i} &\dots &0\\
        \vdots &\ddots &\ddots &\vdots \\
        C_rA_r^{T-2}B_{\mathcal{G},i} &\dots &\dots &C_rB_{\mathcal{G},i} 
    \end{bmatrix}, 
    ~\mathcal{Z}_{i_2,T} = \begin{bmatrix}
        0 &0 &\dots &0\\
        C_rB_{\varpi,i} &0 &\dots &0\\
        C_rA_rB_{\varpi,i} &C_rB_{\varpi,i} &\dots &0\\
        \vdots &\ddots &\ddots &\vdots \\
        C_rA_r^{T-2}B_{\varpi,i} &\dots &\dots &C_rB_{\varpi,i},
    \end{bmatrix}, \\
    &\mathcal{O}_{i,T}=\begin{bmatrix}
        C^{\top}_r & (C_rA_r)^{\top} &\dots &(C_rA_r^{T-1})^{\top}
    \end{bmatrix}^{\top}.
\end{align*}

To eliminate the effect of unknown state~$x_{\tilde{r}_i}(k-T+1)$ on the residual, we further introduce the orthogonal projection of $\mathcal{O}_{i,T}$, i.e., 
$ \mathcal{O}^{\bot}_{i,T} = I-\mathcal{O}_{i,T}\mathcal{O}^{\dag}_{i,T}$.
Multiplying both sides of~\eqref{eq: PS1} from the left by $\mathcal{O}^{\bot}_{i,T}$ leads to
\begin{align}\label{eq: parity-space relation}
\Upsilon_{i,T}(k) = \Psi_{i,T}(k-1) P_{i} + \Omega_{i,T}(k-1),
\end{align}
where $\Upsilon_{i,T}(k) = \mathcal{O}^{\bot}_{i,T} \BDriT(k) \in {\Bbb R}^{n_{\Upsilon}}$,
\begin{align*}
    &\Psi_{i,T}(k-1) = \mathcal{O}^{\bot}_{i,T}\mathcal{Z}_{i_1,T} \BDViT(k-1),\\ &\Omega_{i,T}(k-1) = \mathcal{O}^{\bot}_{i,T}\mathcal{Z}_{i_2,T} \BDWiT(k-1).
\end{align*}
For notational simplicity, the time index is omitted hereafter, e.g.,~$\Upsilon_{i,T}(k) \rightarrow \Upsilon_{i,T}$.

Given~\eqref{eq: parity-space relation}, the weighted LS method is employed here to estimate~$P_{i}$. 
The estimated value~$\hat{P}_{i}$ is given by 
\begin{align}\label{eq: Pli_est}
    \hat{P}_{i} = \arg\min_{P_{i}} \left\|\Upsilon_{i,T}- \Psi_{i,T}P_{i}\right\|^2_{\Sigma^{-1}_{\Omega_{i,T}}}  
    = \Phi_{i,T} \Upsilon_{i,T},
\end{align}
where $\Phi_{i,T} = (\Psi^{\top}_{i,T} \Sigma^{-1}_{\Omega_{i,T}} \Psi_{i,T} )^{-1} \Psi^{\top}_{i,T} \Sigma^{-1}_{\Omega_{i,T}}$ 
and $\Sigma_{\Omega_{i,T}} = \mathcal{O}^{\bot}_{i,T}\mathcal{Z}_{i_2,T}\Sigma_{\BDWiT}(\mathcal{O}^{\bot}_{i,T}\mathcal{Z}_{i_2,T})^{\top} \succ 0$ represents the covariance matrix of~$\Omega_{i,T}$.
To detect load changes and power line faults, the following residual is introduced:
\begin{align}\label{eq: r_tilde}
    \tilde{\Upsilon}_{i,T} = \Upsilon_{i,T} - \Psi_{i,T} \hat{P}_{i},
\end{align}
which characterizes the mismatch between $\Upsilon_{i,T}$ and its reconstruction based on~$\hat{P}_i$. 
In the absence of load changes and line faults, $\hat{P}_i$ obtained from~\eqref{eq: Pli_est} is the minimum-variance unbiased estimate of $P_i$~\cite[Section 4.5]{luenberger1997optimization}, and ~$\Tilde{\Upsilon}_{i,T}$ fluctuates around $0$ due to noise. 
When a load change or power line fault occurs, the mismatch between the actual and estimated loads causes~$\Tilde{\Upsilon}_{i,T}$ to deviate from $0$ apparently, thereby indicating the occurrence of load variations or faults.

We further define the $\kappa$-th row of $\tilde{\Upsilon}_{i,T}$ as
$\tilde{\Upsilon}^{[\kappa]}_{i,T} = \Upsilon^{[\kappa]}_{i,T} - \Psi^{[\kappa]}_{i,T} \hat{P}_{i},~\kappa \in \{1,\dots,n_{\Upsilon}\}$.
When there are no load changes and line faults, according to~\cite[Section 4.5]{luenberger1997optimization}, it holds that 
\begin{align*}
   \mathbf{E} \left[\tilde{\Upsilon}^{[\kappa]}_{i,T}\right] = 0, ~\mathbf{Var}\left[\tilde{\Upsilon}^{[\kappa]}_{i,T}\right] = e_\kappa ({\bf I}_{n_{\Upsilon}} -\Psi_{i,T}  \Phi_{i,T}) \Sigma_{\Omega_{i,T}}e^{\top}_\kappa,
\end{align*}
where $e_\kappa$ is a row vector with all entries zero except for the $\kappa$-th entry being $1$.
Based on Chebyshev's inequality, the probability that~$\tilde{\Upsilon}^{[\kappa]}_{i,T}(k)$ lies inside the threshold interval $[-\varepsilon^{[\kappa]}_{i}(k),\varepsilon^{[\kappa]}_{i}(k)]$ satisfies:
\begin{align*}
    &\mathbf{Pr}\left[ \left. \left|  \tilde{\Upsilon}^{[\kappa]}_{i,T}(k) \right|  \leq \varepsilon^{[\kappa]}_{i}(k) \right| \Delta P_{i}=0, \BDfiT=\mathbf{0}\right]
    \geq 1-\frac{1}{\alpha^2},
\end{align*}
where~$\Delta P_{i}$ denotes the step load change, $\BDfiT$ is the stacked data vector of $f_{I,i}$, $\alpha > 1$ is a tunable scalar, and the time-varying bound is
$\varepsilon^{[\kappa]}_{i}(k) = \alpha \sqrt{\mathbf{Var}[\tilde{\Upsilon}^{[\kappa]}_{i,T}(k)]}$.

Based on the above analysis, deviations caused by both load changes and line faults can be detected if any entries in~$\tilde{\Upsilon}_{i,T}$ exceed the pre-defined threshold interval.
Moreover, we show in the following proposition that~$\tilde{\Upsilon}_{i,T}$ exhibits distinct transient behaviors under these two scenarios, which serves as the foundation for distinguishing between step load changes and power line faults.

\begin{Prop}[Discrimination between step load changes and power line faults]\label{prop: diff_P_fl}
    Consider the power line dynamics~\eqref{line}, the closed-loop dynamics of DG unit $i$~\eqref{eq: Faulty DG}, and Assumptions~\ref{as: uncertainty_noise} and~\ref{as: PWC d}. 
    Suppose that load changes and line faults do not occur simultaneously within the sliding window~$T$. 
    The residual~$\tilde{\Upsilon}_{i,T}$ in~\eqref{eq: r_tilde}, generated through the diagnosis block~$\mathcal{D}_{l,i}$, exhibits distinct patterns under the two scenarios of step load changes and power line faults, as follows:
    \begin{itemize}
        \item[(1)] \textbf{Step load changes in $P_{i}$.} Entries in $\tilde{\Upsilon}_{i,T}$ can temporarily exceed the threshold interval $[-\varepsilon^{[\kappa]}_{i}(k),\varepsilon^{[\kappa]}_{i}(k)]$ due to a step load change. 
        Then, with probability greater than $1-1/ \alpha^2$, the entries return to and remain within the threshold bounds within $T$ steps after the load changes.
        
        \item[(2)] \textbf{Power line faults on DG unit $i$.} The expected value of $\tilde{\Upsilon}^{[\kappa]}_{i,T}$ in the presence of $f_{I,i}$ becomes
        \begin{align*}
     \mathbf{E}\left[\tilde{\Upsilon}^{[\kappa]}_{i,T}\right] = e_\kappa({\bf I}_{n_{\Upsilon}}-\Psi_{i,T} \Phi_{i,T})  \mathcal{O}^{\bot}_{i,T}  \mathcal{Z}_{i_1,T} \BDfiT,
        \end{align*}
        which remains nonzero as long as the fault persists. 
        Then, power line faults can be distinguished from step changes in~$P_{i}$ if there consistently exists at least one entry in $\tilde{\Upsilon}_{i,T}(k)$ that satisfies $\tilde{\Upsilon}^{[\kappa]}_{i,T}(k) \notin [-\varepsilon^{[\kappa]}_{i}(k), \varepsilon^{[\kappa]}_{i}(k)]$ for at least $T$ consecutive steps.
    \end{itemize}
    Therefore, by monitoring whether threshold violations are transient (disappearing within $T$ steps) or persistent (lasting at least $T$ consecutive steps), step load changes and line faults can be discriminated over a time window.
\end{Prop}

\begin{proof}
    The proof is relegated to Appendix~\ref{app: A}.
\end{proof}

To calculate the duration time during which not all~$\tilde{\Upsilon}^{[\kappa]}_{i,T}$ remain within the threshold intervals, we introduce $\mathcal{T}_{i}(k)$ to record the most recent time instant at which all~$\tilde{\Upsilon}^{[\kappa]}_{i,T}$ are inside the thresholds, which is
\begin{align}\label{eq: Td}
    \mathcal{T}_{i}(k) := &\max \left\{k^{\prime} \in \mathbb{N}: \tilde{\Upsilon}^{[\kappa]}_{i,T}(k^{\prime}) \in \left[-\varepsilon^{[\kappa]}_{i}(k^{\prime}),\varepsilon^{[\kappa]}_{i}(k^{\prime}) \right], \forall \kappa \in \{1,\dots,n_{\Upsilon}\},  k 
    \geq k^{\prime} \right\}.
\end{align}
In other words, $k-\mathcal{T}_{i}(k)$ measures the time elapsed since there exists $\tilde{\Upsilon}^{[\kappa]}_{i,T}$ crossing the thresholds.
In addition, we define $\sigma_i(k) \in \{0,1,2\}$ as the status indicator, whose value is determined by the following rules:

\begin{enumerate}
    \item \textbf{No load changes or power line faults ($\sigma(k)=0$)}. If $k-\mathcal{T}_{i}(k) < T/2$, then all entries of~$\tilde{\Upsilon}_{i,T}(k)$ are within their respective thresholds, i.e., $\tilde{\Upsilon}^{[\kappa]}_{i,T}(k) \in [-\varepsilon^{[\kappa]}_{i}(k),\varepsilon^{[\kappa]}_{i}(k) ]$, $\forall \kappa \in \{1,\dots,n_{\Upsilon} \}$, or there exist entries of $\tilde{\Upsilon}_{i,T}(k)$ exceeding thresholds while the duration is less than $T/2$. 
    Note that some entries of~$\tilde{\Upsilon}_{i,T}(k)$ may transiently exceed the thresholds due to stochastic noise even without load changes and faults. 
    To mitigate false alarms, we set $\sigma_i(k)=0$ if threshold violations persist for less than a predefined duration, e.g., $T/2$ here;
    \item \textbf{Step load changes or line faults ($\sigma(k)=1$)}. If $T/2 \leq k-\mathcal{T}_{i}(k) < T$, then there exists at least one element $\tilde{\Upsilon}^{[\kappa]}_{i,T}(k^{\prime}) \notin [-\varepsilon^{[\kappa]}_{i}(k^{\prime}),\varepsilon^{[\kappa]}_{i}(k^{\prime})]$ for each $k^{\prime} \in \{\mathcal{T}_{i}(k)+1,\dots,k\}$ with duration less than $T$ but greater than or equal to $T/2$.
    This corresponds to an intermediate state, in which an event has been detected but not distinguished yet. 
    The discrimination result is obtained by further monitoring whether the threshold violation persists up to $T$ steps or not;
    \item \textbf{Persistent line faults ($\sigma(k)=2$)}. If $k-\mathcal{T}_{i}(k) \geq T$, then there exists $\kappa \in \{1,\dots,n_{\Upsilon}\}$ such that $\tilde{\Upsilon}^{[\kappa]}_{i,T}(k) \notin [-\varepsilon^{[\kappa]}_{i}(k), \varepsilon^{[\kappa]}_{i}(k)]$ for at least $T$ consecutive steps, indicating the occurrence of a line fault.
\end{enumerate}

For clarity, the diagnosis rules are summarized as follows:
\begin{align}\label{eq: disgnosis_rules}
    \sigma_i(k) = \left\{ \begin{array}{ll}
      0, &\text{if}~k-\mathcal{T}_{i}(k) < T/2,  \\
      1, &\text{if}~T/2 \leq k-\mathcal{T}_{i}(k) < T, \\
      2, &\text{if}~k-\mathcal{T}_{i}(k) \geq T.
    \end{array} \right.
\end{align}

Note that the differentiation strategy is developed under the piecewise-constant power load assumption. However, load generally exhibits non-ideal step dynamics in practice, such as start-up transients of electric-vehicle chargers. Nonetheless, the proposed diagnosis approach remains applicable, although the theoretical guarantees established in Proposition~\ref{prop: diff_P_fl}, such as the transient effect on the residual vanishing within $T$ steps, no longer strictly hold. 
In particular, slowly varying load transients may not trigger the detection threshold as $\Delta P$ is small, while fast transients are more likely to be detected since they resemble step-like changes. In both cases, the proposed line-fault estimation performance is maintained, which has been demonstrated by the simulation provided in the Appendix.

\begin{Rem}[Selection of $T$]
    There are some principles in selecting $T$. 
    First, $T$ must satisfy $T>d_{\mathcal{N}}+1$ to ensure the existence of the orthogonal complement of $\mathcal{O}_{i,T}$ in~\eqref{eq: parity-space relation}. 
    Second, $T$ should not be chosen too short, since slowly varying faulty line currents may then appear approximately constant over the window and become difficult to detect because the effect is not sufficiently accumulated over that time interval. On the other hand, a larger $T$ generally improves the detectability of slowly varying faults and robustness to noise, but also increases diagnosis delay and computational burden.
    Therefore, $T$ should be chosen as a compromise among detectability, delay, and complexity.
\end{Rem}

\subsubsection{Estimation of aggregated faulty line current}
Upon successful detection of any faulty line current $f_{I,i}$ in DG unit $i$, an estimator is activated.
The design method of the estimator is provided in this part.
We begin by reformulating the \mbox{parity-space} equation~\eqref{eq: parity-space relation} based on~\eqref{eq: ss PF est} to account for power line faults and obtain:
\begin{align}\label{eq: parity-space relation with fault}
     \Upsilon_{i,T} =\mathcal{O}^{\bot}_{i,T}\mathcal{Z}_{i_1,T} (\BDViT P_{i} + \Delta \boldsymbol{P}_{i} +  \BDfiT)+ \Omega_{i,T},
\end{align}
where 
$\Delta \boldsymbol{P}_{i} =  
\begin{bmatrix}
        \mathbf{0}^{\top} & \boldsymbol{\mathcal{V}}^{\top}_{i,k-k_0}
\end{bmatrix}^{\top}\Delta P_{i}$ 
represents the step load change occurring at $k_0 \in [k-T+1,k-1]$ and $\BDfiT$ is the stacked data vector of $f_{I,i}$.
It is worth noting that $\mathcal{O}^{\bot}_{i,T}\mathcal{Z}_{i_1,T}\BDViT$ is a linear combination of columns of~$\mathcal{O}^{\bot}_{i,T}\mathcal{Z}_{i_1,T}$ and the following rank condition holds:
\begin{align*}
    \text{Rank}\left(\begin{bmatrix}
        \mathcal{O}^{\bot}_{i,T}\mathcal{Z}_{i_1,T} &\mathcal{O}^{\bot}_{i,T}\mathcal{Z}_{i_1,T}\BDViT
    \end{bmatrix}\right) = \text{Rank}(\mathcal{O}^{\bot}_{i,T}\mathcal{Z}_{i_1,T}).
\end{align*}
Consequently, $P_{i}$ and $\BDfiT$ cannot be uniquely determined through~\eqref{eq: parity-space relation with fault}.
To address this issue, instead of directly estimating~$\BDfiT$, we opt to estimate its mean value over the past $T-1$ steps, i.e., 
\begin{align*}
    \BDfiT = \mathbf{1} \bar{f}_{I_i,T-1} + \Delta \boldsymbol{f}_{I_i,T-1},
\end{align*}
where $\bar{f}_{I_i,T-1} = \frac{1}{T-1}\sum^{k-1}_{k-T+1} f_{I,i}(k^{\prime})$ is the mean value, $\Delta \boldsymbol{f}_{I_i,T-1}$ represents the deviation, and $\mathbf{1}=[1 ~\dots ~1]^{\top}$.
Then, the parity-space relation~\eqref{eq: parity-space relation with fault} is reformulated as:
\begin{align}\label{eq: Reform_PS}
\begin{split}
    \Upsilon_{i,T}  = \Gamma_{i,T}\Theta_i
     +\mathcal{O}^{\bot}_{i,T} \mathcal{Z}_{i_1,T} (\Delta {\boldsymbol{f}}_{I_i,T-1} +\Delta \boldsymbol{P}_{i} )
     + \Omega_{i,T}.
\end{split}
\end{align}
where 
$\Gamma_{i,T} = \mathcal{O}^{\bot}_{i,T} \mathcal{Z}_{i_1,T} [\BDViT ~\mathbf{1}]$ and
$\Theta_i =  [P_{i} ~\bar{f}_{I_i,T-1}]^{\top}$.

In~\eqref{eq: Reform_PS}, Rank$(\Gamma_{i,T}) = 2$ is of full column rank as long as $\BDViT$ is not a constant vector. 
Nevertheless, the estimate of~$\Theta_i$ is still susceptible to load deviations~$\Delta \boldsymbol{P}_{i}$, faulty current fluctuations~$\Delta {\boldsymbol{f}}_{I_i,T-1}$, and noise term~$\Omega_{i,T}$.
This sensitivity stems from the ill-conditioning of $\Gamma_{i,T}$ when~$\BDViT$ exhibits near-constant behavior.
Therefore, to estimate $\bar{f}_{I_i,T-1}$ robustly, we formulate the following regularized LS problem:
\begin{align}\label{eq: QP}
    \hat{\Theta}_i =  \arg \min_{\Theta_i} ~\|\Upsilon_{i,T} - \Gamma_{i,T} \Theta_i\|^2_{\Sigma^{-1}_{\Omega_{i,T}}} + \eta \|\nu_1(\Theta_i - \hat{\Theta}_{i\_})\|_2^2,
\end{align}
where $\hat{\Theta}_i = [\hat{P}_i ~\hat{f}_{I,i}]^{\top}$ is the estimate of~$\Theta_i$,
$\eta \in \mathbb{R}_+$ is the regularization weight,~$\nu_1=[1 ~0]$ selectively penalizes variations in the load estimate, and $\hat{\Theta}_{i\_}$ represents the estimation result from the previous time step.
The regularization term enforces temporal consistency on~$P_{i}$, reflecting its piece-wise constant nature in practical systems. 

Although~\eqref{eq: QP} is a quadratic optimization problem and can be solved tractably, it is not suitable for online monitoring as it requires solving the optimization problem at each step, which can be computationally expensive. 
Fortunately,~\eqref{eq: QP} admits an analytical solution, as presented in the following proposition.

\begin{Prop}[Analytical solution]\label{prop: analytical sol}
The analytical solution to the quadratic optimization problem~\eqref{eq: QP} is given by
\begin{align}\label{eq: analytical sol}
    \hat{\Theta}_i= \mathcal{K}^{-1}_{i,T}
    \left(\Gamma_{i,T}^{\top}\Sigma^{-1}_{\Omega_{i,T}} \Upsilon_{i,T}+\eta \nu_1^{\top}\nu_1 \hat{\Theta}_{i\_}\right).
\end{align}
where $\mathcal{K}_{i,T} = \Gamma_{i,T}^{\top}\Sigma^{-1}_{\Omega_{i,T}}\Gamma_{i,T}+\eta \nu_1^{\top}\nu_1$. 
The estimate of $\bar{f}_{I_i,T-1}$ is obtained by~$\hat{f}_{I,i} = \nu_2 \hat{\Theta}_i$ with~$\nu_2=[0 ~1]$.
\end{Prop}

\begin{proof}
    By defining the objective function in \eqref{eq: QP} as $J(\Theta_i)$ and taking the derivative of $J(\Theta_i)$ over $\Theta_i$, we have
    \begin{align*}
        \frac{\partial J(\Theta_i)}{\partial \Theta_i} =& 2\left(\Gamma^{\top}_{i,T}\Sigma^{-1}_{\Omega_{i,T}}\Gamma_{i,T}+\eta \nu_1^{\top}\nu_1 \right)\Theta_i-        2\left(\Gamma^{\top}_{i,T}\Sigma^{-1}_{\Omega_{i,T}} \Upsilon_{i,T} +\eta \nu_1^{\top}\nu_1 \hat{\Theta}_{i\_} \right).
    \end{align*}
    The solution is obtained by setting $\partial J(\Theta_i) / \partial \Theta_i = 0$. This completes the proof.
\end{proof}

\begin{algorithm}[t] 
	\caption{Fault diagnosis process for DG unit $i$} \label{alg:algorithm_1} 
	\begin{algorithmic}
	\State \textbf{I. Estimation of the actuator fault $f_{a,i}$} 
	\begin{itemize}
	    \item[(a)] Input $Y_i$ to the estimation filter~\eqref{eq: filter}  
	    \item[(b)] Output the estimated value $\hat{f}_{a,i}$
	\end{itemize}
	\State \textbf{II. Estimation of the faulty line current $f_{I,i}$}
	\begin{itemize}
	    \item[(a)] \textit{Initialization}: Choose the time window length $T$ and initialize the iteration index $k=T$. Set the status indicator $\sigma_{i}(k-1)=0$ and the crossing time recorder $\mathcal{T}_{i}(k-1)=k-1$
	    \item[(b)] Generate $\tilde{r}_i$ using pre-filter 1  in~\eqref{eq: filter fl}, line current estimators in~\eqref{eq: Ik_est}, pre-filter 2 in~\eqref{eq: prefilter 2}, and~\eqref{eq: ss PF est}
            \item[(c)] Generate $\Upsilon_{i,T} = \mathcal{O}^{\bot}_{i,T} \BDriT$ in~\eqref{eq: parity-space relation} with discretized~$\tilde{r}_i$ 
            \item[(d)] Compute the covariance matrix $\Omega_{i,T}$
            \item[(e)] If $\sigma_{i}(k-1)=0$ or $\sigma_{i}(k-1)=1$
                         \% Fault-load discrimination phase:
	          \begin{itemize}
	              \item[] $\hat{P}_{i} = \Phi_{i,T}\Upsilon_{i,T}$ in~\eqref{eq: Pli_est};
                    \item[] $\tilde{\Upsilon}_{i,T} = \Upsilon_{i,T} - \Psi_{i,T} \hat{P}_{i}$ in~\eqref{eq: r_tilde};
                    \item[] If $\exists \tilde{\Upsilon}^{[\kappa]}_{i,T} \notin \left[-\varepsilon^{[\kappa]}_{i}(k),\varepsilon^{[\kappa]}_{i}(k)\right]$
                    \begin{itemize}
                        \item[]  $\mathcal{T}_{i}(k) = \mathcal{T}_{i}(k-1)$; \% Record crossing time
                    \end{itemize}
                    \item[] Else 
                    \begin{itemize}
                        \item[]  $\mathcal{T}_{i}(k) = k$; 
                    \end{itemize}
                    \item [] End
                    \item[] Use the diagnosis rules~\eqref{eq: disgnosis_rules} to determine the  current value of the indicator $\sigma_i(k)$;
	    \end{itemize}
	    \item[(f)] Else \%Estimation phase: 
                \begin{itemize}
                    \item[] Generate the estimated value $\hat{\Theta}_i$ using~\eqref{eq: analytical sol};
                    \item[] Set $\sigma_{i}(k) = 2$;
                \end{itemize}
            \item[(g)] End
            \item[(h)] Set $k=k+1$ and repeat (b)-(h).
	\end{itemize}
	\end{algorithmic} 
\end{algorithm}

The diagnosis processes of both actuator and power line faults in DG unit $i$ are summarized in Algorithm~\ref{alg:algorithm_1}. 
Furthermore, to facilitate the subsequent result about the performance bound of the estimator~\eqref{eq: analytical sol}, some notations are introduced. 
Recall that 
$\Psi_{i,T} = \mathcal{O}^{\bot}_{i,T}\mathcal{Z}_{i_1,T}\BDViT$ in~\eqref{eq: parity-space relation}. 
We further define~$\bar{\mathcal{Z}}_{i,T} = \mathcal{O}^{\bot}_{i,T}\mathcal{Z}_{i_1,T} \mathbf{1}$, thus $\Gamma_{i,T}$ in~\eqref{eq: Reform_PS} becomes~$\Gamma_{i,T} = [\Psi_{i,T} ~\bar{\mathcal{Z}}_{i,T}]$.
Let $\underline{\lambda}_{\Sigma_{\Omega_{i,T}}}$ and $\bar{\lambda}_{\Sigma_{\Omega_{i,T}}}$ denote the smallest and largest eigenvalues of~$\Sigma_{\Omega_{i,T}}$, respectively.
The following theorem provides a bound for the estimation error of the average value of~$f_{I,i}$.

\begin{Thm}[Performance bound of the regularized estimator]\label{Thm}
Considering the power line dynamics~\eqref{line}, the closed-loop dynamics of DG unit $i$~\eqref{eq: Faulty DG}, and Assumptions~\ref{as: uncertainty_noise} and~\ref{as: PWC d}, the estimated value~$\hat{f}_{I,i}$ of $\bar{f}_{I_i,T-1}$ obtained from~\eqref{eq: analytical sol} through the diagnosis block~$\mathcal{D}_{l,i}$, satisfies the following error bound:
\begin{align}\label{eq: error bound}
        \left|\mathbf{E}\left[\bar{f}_{I_i,T-1} -  \hat{f}_{I,i}\right] \right| &\leq \frac{\bar{\lambda}_{\Sigma_{\Omega_{i,T}}}}{\underline{\lambda}_{\mathcal{M}_{i,T}}+\eta \bar{\lambda}_{\Sigma_{\Omega_{i,T}}}} 
        \frac{\sqrt{ \bar{\lambda}_{\mathcal{M}_{i,T}}}} {\underline{\lambda}_{\Sigma_{\Omega_{i,T}}}} 
       \bar{\sigma}(\mathcal{O}^{\bot}_{i,T} \mathcal{Z}_{i_1,T}) 
       \left( \|\Delta \boldsymbol{f}_{I_i,T-1} \|_2 + \| \Delta \boldsymbol{P}_{i}  \|_2 \right) \notag\\
       &+{\frac{\bar{\lambda}_{\Sigma_{\Omega_{i,T}}} \bar{\mathcal{Z}}^{\top}_{i,T}  \Psi_{i,T}}       {\underline{\lambda}_{\Sigma_{\Omega_{i,T}}}\bar{\mathcal{Z}}^{\top}_{i,T}  \bar{\mathcal{Z}}_{i,T}} }
        | P_{i} - \hat{P}_{i\_} |,
    \end{align}
    where~$\bar{\lambda}_{\mathcal{M}_{i,T}} = \Psi^{\top}_{i,T} \Psi_{i,T} + \bar{\mathcal{Z}}^{\top}_{i,T} \bar{\mathcal{Z}}_{i,T}$ and 
    \begin{align*}
        \underline{\lambda}_{\mathcal{M}_{i,T}} =\frac{ \Psi^{\top}_{i,T} \Psi_{i,T} \bar{\mathcal{Z}}^{\top}_{i,T} \bar{\mathcal{Z}}_{i,T} - \Psi^{\top}_{i,T}\bar{\mathcal{Z}}_{i,T} \bar{\mathcal{Z}}^{\top}_{i,T} \Psi_{i,T}}{\Psi^{\top}_{i,T} \Psi_{i,T} + \bar{\mathcal{Z}}^{\top}_{i,T} \bar{\mathcal{Z}}_{i,T}}.
    \end{align*} 
\end{Thm}

\begin{proof}
    The proof is relegated to Appendix~\ref{app: B}.
\end{proof}

The bound in~\eqref{eq: error bound} holds for each time instant $k$ after power line fault detection, and for simplicity, $k$ is omitted here. 
Let us elaborate further on how different factors influence the bound:
\begin{enumerate}
    \item \textbf{Ill-conditioning issue.} When~$V_i(k)$ approaches a constant value, $\underline{\lambda}_{\mathcal{M}_{i,T}}$ in the denominator of~\eqref{eq: error bound} tends to zero, and thus the bound becomes infinite without the regularized term in~\eqref{eq: QP}, i.e.,~$\eta = 0$.
    This makes the estimation problem ill-conditioned. 
    
    \item \textbf{Variation terms $\Delta \boldsymbol{f}_{I_i,T-1}$, $\Delta \boldsymbol{P}_{i}$, and the noise level.} The tightness of the bound is mainly affected by $\Delta \boldsymbol{f}_{I_i,T-1}$, $\Delta \boldsymbol{P}_{i}$, and $\bar{\lambda}_{\Sigma_{\Omega_{i,T}}}$. 
    Particularly, when the faulty current is constant (or converges to a constant) and the load remains unchanged, the bound is tight. In contrast, it becomes conservative when the faulty current continues to evolve or when step load changes occur.  The noise level affects the bound through ~$\bar{\lambda}_{\Sigma_{\Omega_{i,T}}}$. 
    Therefore, a larger~$\eta$ is generally needed to maintain a relatively tight bound when the noise is smaller, and to mitigate the effects of non-zero $\Delta \boldsymbol{f}_{I_i,T-1}$ and $\Delta \boldsymbol{P}_{i}$.
    
    \item \textbf{Priori estimation result~$\hat{P}_{i \_}$.} 
    The regularization term in~\eqref{eq: QP} enforces $P_{i}$ to stay close to its prior estimate. 
    Initially, $|P_{i} - \hat{P}_{i\_}|$ is small as the estimation result of~$P_{i}$ from differentiation phase~\eqref{eq: Pli_est} is unbiased. 
    However, step changes in~$P_{i}$ during the estimation phase will increase the estimation error, which is explicitly reflected in the last term of~\eqref{eq: error bound}.
    Moreover, as~$\eta$ approaches infinity, the error bound converges to the value determined by~$|P_{i} - \hat{P}_{i\_}|$.
\end{enumerate}

\begin{Rem}[Selection of $\eta$]
The selection of~$\eta$ mainly depends on the ill-conditioning of~$\Gamma_{i,T}$, the noise level characterized by $\Sigma_{\Omega_{i,T}}$, and the bias--robustness trade-off. In general, a larger~$\eta$ is needed when $\Gamma_{i,T}$ is more ill-conditioned or when the data-fitting term is weighted more heavily. Increasing~$\eta$ improves robustness to variations in $P_i$ and $f_{I_i}$, but may introduce bias if the prior load estimate is inaccurate. We show in the simulation that increasing $\eta$ can reduce the estimation error.
\end{Rem}

% \begin{Rem}[Location of the faulty line]
%    Faults on a power line can be detected and estimated by the diagnosis components of the DG units located at both ends of the line. 
%    By analyzing the estimated faulty line currents, the faulty line can be identified in some scenarios. 
%    For instance, in the simplest case where only one line in the DC microgrid is faulty, the fault location can be precisely determined.
% \end{Rem}

\begin{Rem}[Applicability range of faults and possible extension]
    Except for the actuator and power line faults, the proposed diagnosis framework also applies to other faults that can be modeled as additive unknown inputs with known structural directions that do not lie in the subspace spanned by the considered disturbances, such as sensor faults and parameter faults. 
    However, the current design does not explicitly isolate these additional fault types. Fault isolation could be achieved by introducing additional diagnosis filters with structured fault sensitivity.
\end{Rem}

\begin{Rem}[Complexity analysis]
    Although the faulty line current estimation algorithm involves multiple computational steps, its complexity is polynomial in the filter order~$d_{\mathcal{N}}$, the input dimension~$n_y$, and the projected residual dimension~$n_{\Upsilon}$, and it does not require iterative online optimization. Consequently, the proposed method incurs a relatively low online computational burden. A detailed analysis of the computational complexity is provided in the Appendix.
    To further demonstrate the real-time feasibility of the proposed algorithm, we also evaluate its average execution time in simulation.
\end{Rem}

%%%%%%%%%%%%%%%%%%%%%%%%%%%%%%%%%%%%%%%%%%%%%%%%%%%%%%%%%%%%%%%%%%%%%%
%%%%%%%%%%%%%%%%%%%%%%%%%%%%%%%%%%%%%%%%%%%%%%%%%%%%%%%%%%%%%%%%%%%%%
%                   Simulation results
%%%%%%%%%%%%%%%%%%%%%%%%%%%%%%%%%%%%%%%%%%%%%%%%%%%%%%%%%%%%%%%%%%%%%%%
\section{Simulation Results}\label{sec: 5}
In this section, the effectiveness of the proposed diagnosis scheme is validated on the DC microgrid depicted in Fig.~\ref{fig: DCMG}~\cite{sadabadi2023resilient}.  
All simulations are performed using MATLAB R2022b on a workstation equipped with an Intel i7-13800H processor (2.50 GHz) and 32 GB RAM. The optimization problems involved in the proposed method are solved using the YALMIP toolbox~\cite{lofberg2004yalmip}.
The parameters of the microgrid are provided in Table \ref{tab:init}, with the reference voltages set as $V^*_1=48$~V, $V^*_2=48.1$~V, and $V^*_3 = 47.5$~V.
Suppose that the microgrid is initially working under normal conditions. 
Then, faults might occur when the system is at or close to steady state.
The initial conditions are thus set as 
$x_1(0) = [47.80, 6.93, 11.19]^{\top}, x_2(0) = [48.10, 4.76, 15.61]^{\top}, \text{and} ~x_3(0) = [47.50, -4.29, 15.02]^{\top}.$
Additionally, since a discretization approach is employed to estimate the faulty current $f_{I,i}$, all simulation results are presented using sampling data to ensure consistency, with a sampling time of $t_s = 1\times10^{-5}$ s.
Based on the above setup, the diagnosis performance of $\mathcal{D}_{a,i}$ and $\mathcal{D}_{l,i}$ is evaluated separately in the subsequent subsections.

\begin{table}[t]
\begin{small}
\begin{center}
\caption{\small Parameters of the DC microgrid system.} \label{tab:init}   
        \begin{tabular}{ cc |cc}
        \hline 
        Name &Values &Name &Values\\
        \hline
        $R_{t,1}$  &$0.2 ~\si{\ohm}$     &$L_{t,1}$ &$1.8$ mH\\
        $R_{t,2}$  &$0.3 ~\si{\ohm}$     &$L_{t,2}$ &$2$ mH \\
        $R_{t,3}$  &$0.1 ~\si{\ohm}$     &$L_{t,3}$ &$2.2$ mH\\   
        $C_{t,1}$  &$2.21$ mF            &$R_1$  &0.05 $\si{\ohm}$\\
        $C_{t,2}$  &$1.9$ mF             &$R_2$  &0.07 $\si{\ohm}$\\
        $C_{t,3}$  &$1.7$ mF             &$L_1$  &2.1 $\mu$F\\
        $K_1$   &$[-15 ~-2 ~70]^{\top}$   &$L_2$  &1.8 $\mu$F\\
        $K_2$   &$[-15 ~-2 ~50]^{\top}$   &$K_3$   &$[-15 ~-2 ~50]^{\top}$  \\
        \hline
        \end{tabular}
  \end{center}
\end{small}
\end{table}

\subsection{Actuator Fault Diagnosis Results}
To estimate actuator faults, an estimator $\mathcal{D}_{a,i}$ for each DG unit is designed by utilizing the approach proposed in Section~\ref{sec: 3}. 
To demonstrate the advantages of the estimation filter constructed within the DAE framework, we compare it with a sliding mode observer (SMO)~\cite{edwards2000sliding} and an adaptive observer (AO)~\cite{wang2002actuator}.
Taking~$\mathcal{D}_{a,1}$ as an example, the design process is outlined as follows:

\textit{Step 1}. Determine the degree and denominator of the filter in~\eqref{eq: filter}. For this example, the degree is $d_N = 2$, and the denominator is chosen as $a(p) = (0.5+p)(0.1+p)(1+p)$.

\textit{Step 2}. Obtain the parameters of the filter's numerator by solving the optimization problem~\eqref{eq: opt fa} and construct the filter in~\eqref{eq: filter}.

Suppose that an incipient fault~$f_{a,1}$ occurs in the voltage controller of DG $1$ and a step fault~$f_{a,3}$ occurs in the voltage controller of DG $3$. 
The incipient fault~$f_{a,1}(t) = 0$ for $t \leq 100$~ms, and for~$t \geq 100$~ms, $f_{a,1}(t)$ evolves as:
\begin{align*}
    \dot{f}_{a,1}(t) = -\beta_a f_{a,1}(t) + \beta_a \bar{f}_{a,1},
\end{align*}
where $\beta_a = 8 \times 10^{-9}$ determines the changing rate and stealthiness of~$f_{a,1}$ and $\bar{f}_{a,1} = 5$ is the final value of $f_{a,1}$.
The step fault~$f_{a,3}$ remains $0$ until~$t=70$~ms, after which it becomes~$2$. 

To demonstrate that the designed estimator can decouple the effects of load changes and power line faults, we consider a step change in~$P_{1}$, i.e., $P_1$ steps from $100$ to $120$ at $40$~ms, as well as the occurrence of a power line fault between DG~$1$ and DG~$2$, i.e., $f_{L,1}$ steps from 0 to 0.1 at 80 ms.
The power load~$P_{2}$ changes from $125$~W to $110$~W at $120$~ms, and~$P_{3}(t) = 130$~W is constant.
Moreover, noise $\delta_i$ and $\zeta_i$ for $i=\{1,2,3\}$, $\epsilon_k$ for $k=\{1,2\}$ are zero mean white noise with standard deviations $0.01$, $0.001$, and $0.01$, respectively.

The actuator fault estimation results are presented in Fig.~\ref{fig: VI4Fa} and Fig.~\ref{fig: FaResult}.
Specifically, Fig.~\ref{fig: VI4Fa} depicts the voltage and current in response to load changes and faults. 
As shown, the measured signals $V_i$ and $I_{t,i}$ do not provide clear indications of fault occurrence.
Fig.~\ref{fig: FaResult} presents the diagnosis results from the local actuator fault estimators.
For brevity, we focus on the results of~$\mathcal{D}_{a,1}$ depicted in Fig.~\ref{fig: FaResult}~(a), where the actuator fault $f_{a,1}$ and its estimates generated by~\eqref{eq: filter}, AO, and SMO are compared. 
The results show that the proposed estimate~$\hat{f}_{a,1}$ accurately follows the true fault signal~$f_{a,1}$ and remains unaffected by load changes, power line faults, and neighboring DG unit dynamics. 
By contrast, both AO and SMO methods respond more slowly after the fault occurs, while the SMO-based approach also exhibits unavoidable chattering in the estimation results.

\begin{figure}[t]
    \centering
%    \captionsetup{justification=centering}
    \includegraphics[width=0.4\linewidth]{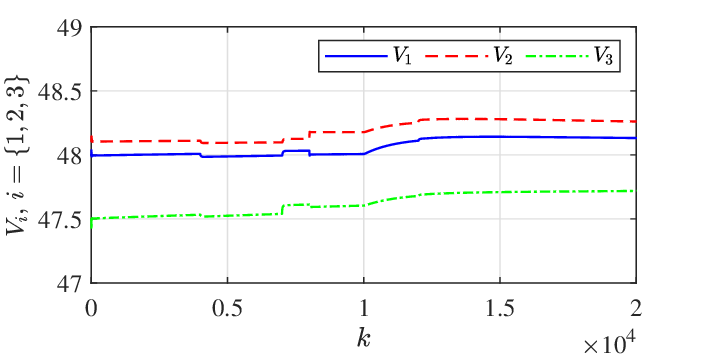}
    \includegraphics[width=0.4\linewidth]{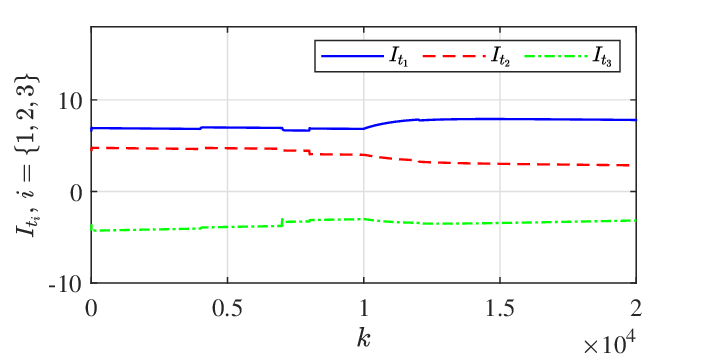}
    \caption{\small Dynamic response of the DC microgrid when considering the actuator fault.}
    \label{fig: VI4Fa}
\end{figure}

% \begin{figure*}[htbp]
%     \centering
%     \subfigure[]{\includegraphics[width=0.32\textwidth]{Figures/Fig_FaResult_DG1.eps}\label{fig: FaResult_DG1}}
%     \hfill
%     \subfigure[]
%     {\includegraphics[width=0.32\textwidth]{Figures/Fig_FaResult_DG2.eps}\label{fig: FaResult_DG2}}
%     \hfill
%     \subfigure[]
%     {\includegraphics[width=0.32\textwidth]{Figures/Fig_FaResult_DG3.eps}\label{fig: FaResult_DG3}}
%     \caption{\small Diagnosis of actuator faults: 
%     (a) diagnosis results of~$\mathcal{D}_{a,1}$,   
%     (b) diagnosis results of~$\mathcal{D}_{a,2}$, and  
%     (c) diagnosis results of~$\mathcal{D}_{a,3}$.}
%     \label{fig: FaResult}
% \end{figure*}

\begin{figure}[t]
    \centering
    \includegraphics[width=0.95\linewidth]{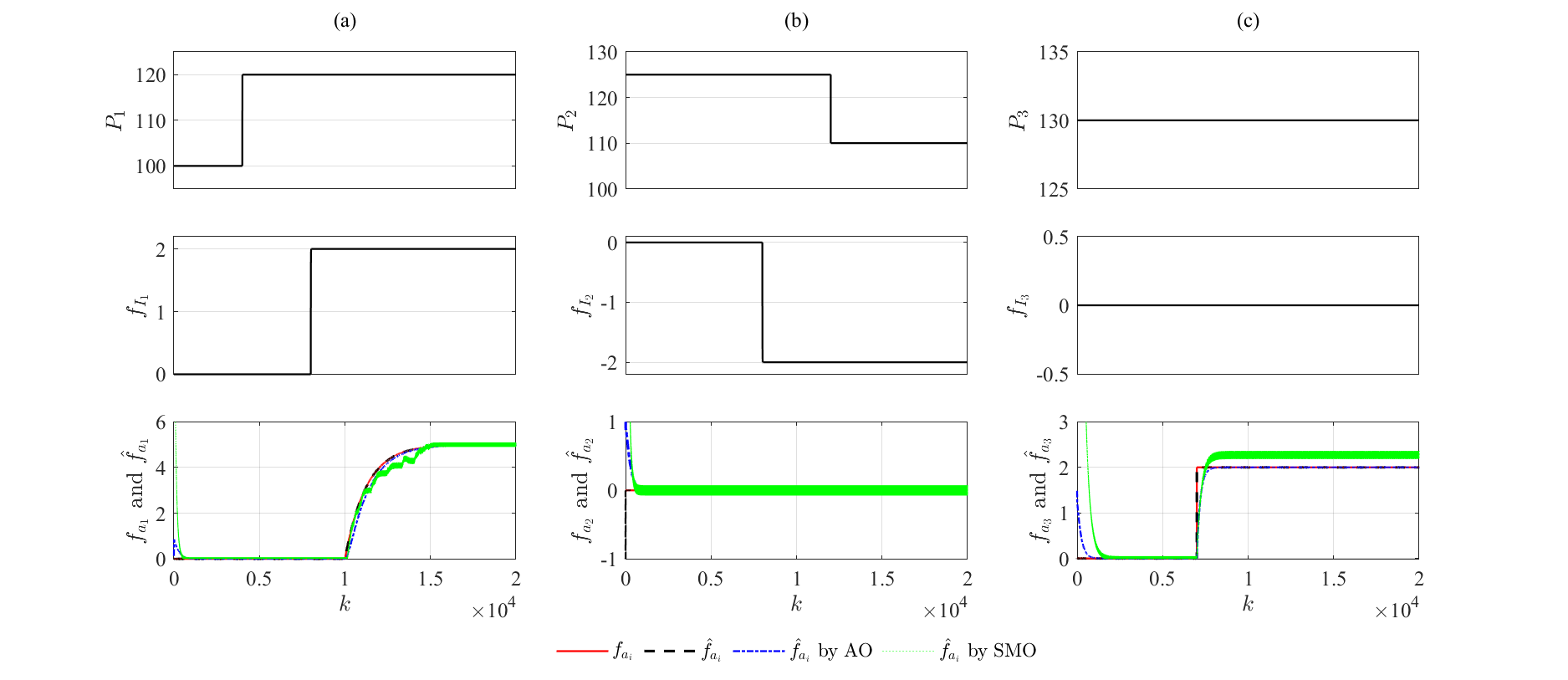}
    \caption{\small Diagnosis of the actuator fault: (a) results of~$\mathcal{D}_{a,1}$, (b) results of~$\mathcal{D}_{a,2}$, and (c) results of~$\mathcal{D}_{a,3}$.}
    \label{fig: FaResult}
\end{figure}

\subsection{Power Line Fault Diagnosis Results}
For the estimation of faulty line current $f_{I,i}$ induced by power line faults, we introduce the block $\mathcal{D}_{l,i}$ for each DG unit, whose structure is illustrated in Fig.~\ref{fig: PF est}. 
The proposed estimation scheme is compared with the multiple fault estimation method developed in~\cite{van2022multiple}.
Moreover, to better capture the behavior of practical microgrid systems, we construct a high-fidelity microgrid model in Simulink/Simscape Electrical using electrical and power electronics components with dynamic characteristics. The proposed diagnosis method is then validated on this model to demonstrate its potential practical applicability.
Taking $\mathcal{D}_{l,1}$ as an example, its design process is summarized as follows:

\textit{Step 1.} Design Pre-filter 1 depicted by~\eqref{eq: filter fl} in Section~\ref{sec: 4}.1.
We choose the degree $d_{\mathcal{N}} = 2$ and the denominator $a(p) = (0.5+p)(0.1+p)(1+p)$. 
The coefficients of $\mathcal{N}_1(p)$ are determined using the approach outlined in Proposition~\ref{prop: fa est}. 
%The residual generator~\eqref{eq: filter fl} takes locally available signals $y_1$ and $V^*_1$ as inputs and produces the output $r_{1}$.

\textit{Step 2.} Design line current estimators described by~\eqref{eq: Ik_est} in Section~\ref{sec: 4}.2 for all power lines connected to DG $1$. 
%These line current estimators use the PCC voltages at the ends of the power lines (i.e.,~$V_1,V_2$, and $V_3$ for DG $1$) as inputs and generate the estimated fault-free component of the line currents~$\hat{I}_{1,h}$ and~$\hat{I}_{2,h}$.

\textit{Step 3.} Design Pre-filter $2$ described by~\eqref{eq: prefilter 2} in Section~\ref{sec: 4}.3 and generate the residual~$\tilde{r}_{1}$ used for faulty line current estimation. 
%The input of Pre-filter $2$ is the sum of the estimated healthy line currents~$\hat{I}_{1,h}$ and~$\hat{I}_{2,h}$.
%Then, the residual~$\tilde{r}_{1}$ can be obtained based on~\eqref{eq: til_ri}.

\textit{Step 4.} Design the faulty line current estimator for $f_{I,1}$ based on~\eqref{eq: Reform_PS}-\eqref{eq: analytical sol}. 
Discretize the system with a sampling time of $t_s = 1\times10^{-5}$ s and choose a sliding window length of $T=20$. 
%The residual $\tilde{r}_{1}$ is first used to distinguish between step load changes and power line faults according to Proposition~\ref{prop: diff_P_fl}. 
The estimator \eqref{eq: analytical sol} is activated when a power line fault is detected, which produces the estimate of~$\bar{f}_{I_1,T-1}$.  

To verify the performance of the designed faulty current estimator~$\mathcal{D}_{l,i}$, we consider two types of line faults: the incipient fault and the pole-to-ground short-circuit fault. 

\noindent\textit{\bf Case I: Incipient fault on power line $1$.}
In the first scenario, we consider an incipient fault occurring on the power line between DG~$1$ and DG~$2$.
The incipient power line fault $f_{L,1}(t)=0$ for $t \leq 80$ ms, and for~$t>80$ ms, its dynamics become:
\begin{align*}
     \dot{f}_{L,1}(t) = -\beta_l f_{L,1}(t) + \beta_l \bar{f}_{L,1}, 
\end{align*}
where $\beta_l = 4 \times 10^{-9}$ and $\bar{f}_{L,1} = 1$.
To further evaluate the decoupling capability of $\mathcal{D}_{l,1}$, an actuator fault~$f_{a,1}$ is considered in DG~$1$, which is modeled as a step signal, i.e., $f_{a,1}(t)=0$ for $t \leq 60$~ms and $f_{a,1}(t)=2$ for $t > 60$~ms.
Step changes in~$P_{1}$ and~$P_{3}$ are considered, i.e.,~$P_{1}(t)$ changes from $100$~W to~$120$~W at~$t = 40$~ms and $P_{3}(t)$ changes from $140$~W to $130$~W at~$t = 120$~ms. 
The load of DG $2$ remains constant at $P_{2}(t)=110$~W. 

\begin{figure}[t]
    \centering
%    \captionsetup{justification=centering}
    \includegraphics[width=0.4\linewidth]{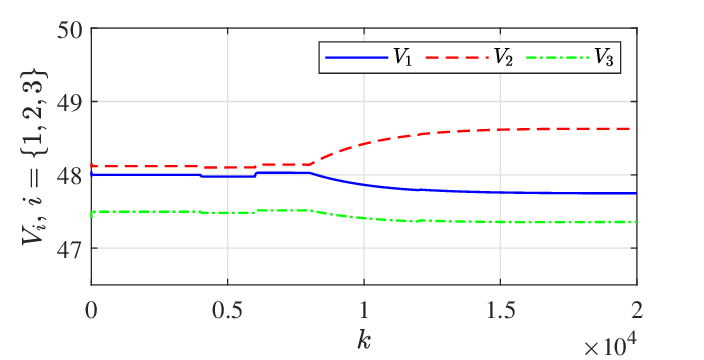}
    \includegraphics[width=0.4\linewidth]{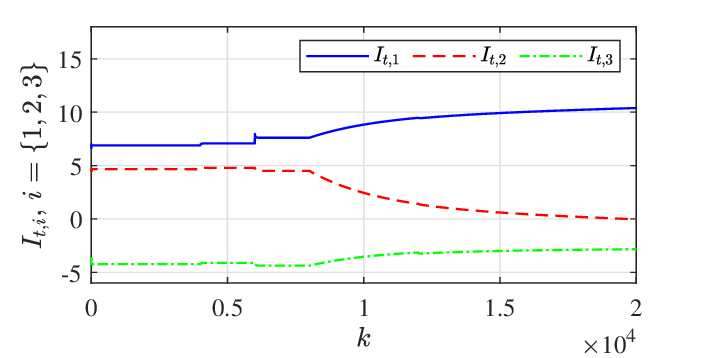}
    \caption{\small Dynamic response of the DC microgrid when considering the incipient power line fault.}
    \label{fig: VI4SFl}
\end{figure}

\begin{figure}[t]
    \centering
    \includegraphics[width=0.95\linewidth]{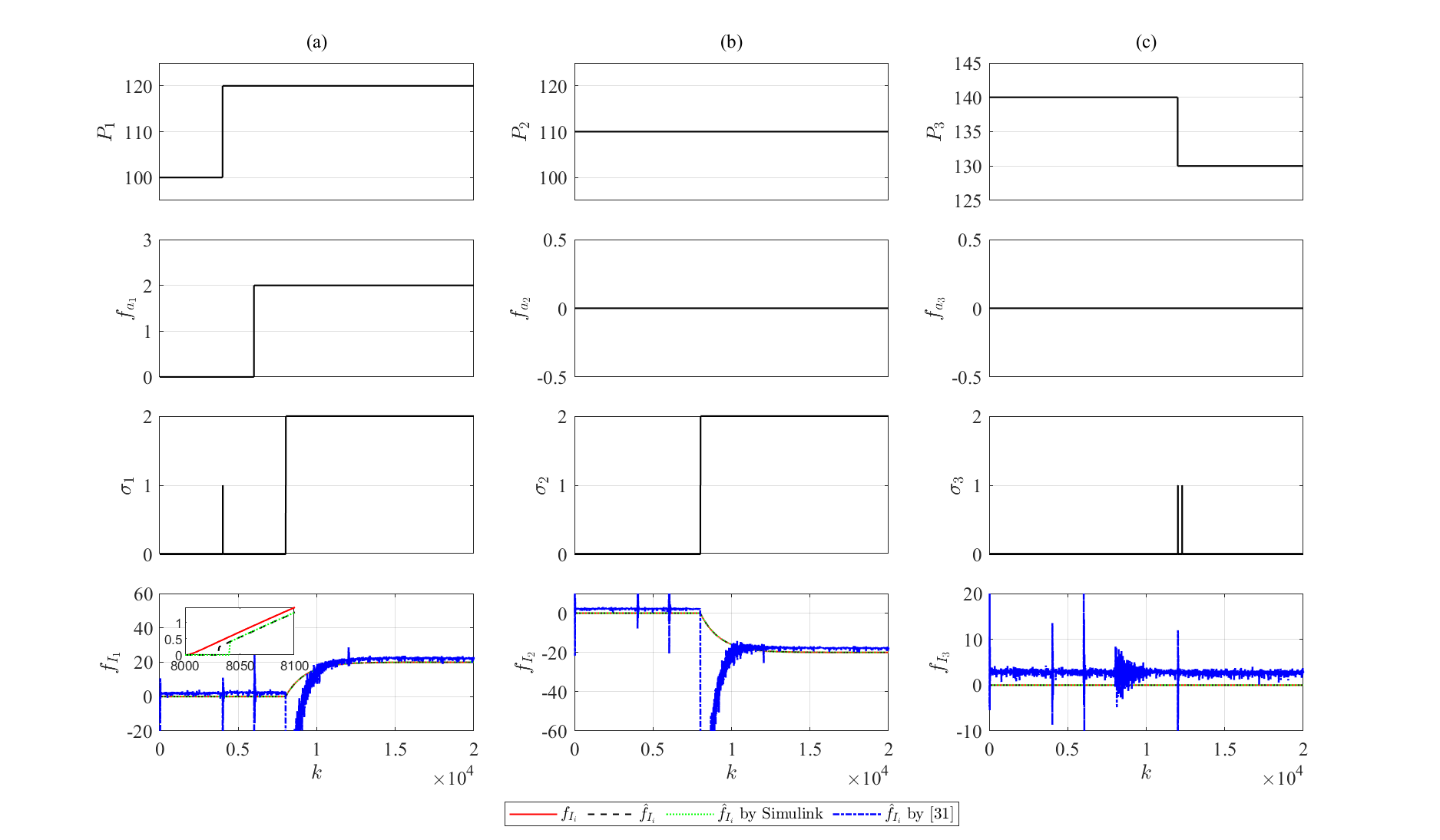}
    \caption{\small Diagnosis of the incipient power line fault: (a) results of~$\mathcal{D}_{l,1}$, (b) results of~$\mathcal{D}_{l,2}$, and (c) results of~$\mathcal{D}_{l,3}$.}
    \label{fig: FlResult}
\end{figure}

Figs.~\ref{fig: VI4SFl}-\ref{fig: PFB4SFl} present the diagnosis results of the incipient power line fault $f_{L,1}$. 
Particularly, Fig.~\ref{fig: VI4SFl} illustrates the voltage and current variations of each DG unit under load changes and faults.
Due to the stealthiness of the fault, the occurrence of~$f_{L,1}$ cannot be discerned directly from these measurements.
Fig.~\ref{fig: FlResult} depicts the diagnosis results of the faulty line current estimator of each DG unit.
For instance, Fig.~\ref{fig: FlResult}(a) shows the results of $\mathcal{D}_{l,1}$.
At $t = 40$~ms, the step change in $P_{1}$ is detected according to the diagnosis rules in~\eqref{eq: disgnosis_rules}, as the status indicator~$\sigma_1$ briefly switches to~$1$ with a duration shorter than~$T$.
At $t = 80$ ms, the fault~$f_{L,1}$ happens and~$\sigma_1$ becomes $2$, meaning the successful detection of the power line fault~$f_{L,1}$.

The bottom sub-figure of Fig.~\ref{fig: FlResult}(a) compares the estimates of the faulty current~$f_{I,1}$ obtained using the proposed DAE-based approach, the Simulink model, and the estimation approach in~\cite{van2022multiple}.
As shown, the proposed estimation approach tracks the faulty current well, whereas the method in~\cite{van2022multiple} exhibits significant estimation errors and oscillations due to near-singularity.
The Simulink-based results are consistent with those obtained from the MATLAB implementation, further supporting the practical applicability of the proposed diagnosis framework.
It is also worth noting that the actuator fault~$f_{a,1}$ and dynamics of other DG units have no effects on the diagnosis results of $\mathcal{D}_{l,1}$.
Additionally, the power line fault~$f_{L,1}$ is also detected and estimated by~$\mathcal{D}_{l,2}$, as shown in Fig.~\ref{fig: FlResult}~(b). 
Fig.~\ref{fig: PFB4SFl} presents the value $|\mathbf{E}[\bar{f}_{I_1,T-1}-\hat{f}_{I,1}]|$ and its upper bound given in~\eqref{eq: error bound}. 
As the faulty current gradually converges to a constant, the upper bound of the error also decreases, which aligns well with the theoretical analysis.
Fig.~\ref{fig: DifferentEta} further illustrates the effect of increasing~$\eta$ in reducing the estimation error.

\begin{figure}[t]
    \centering
    \includegraphics[width=0.4\linewidth]{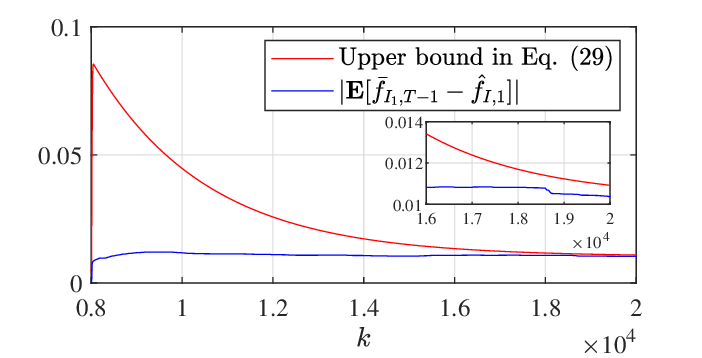}
    \includegraphics[width=0.4\linewidth]{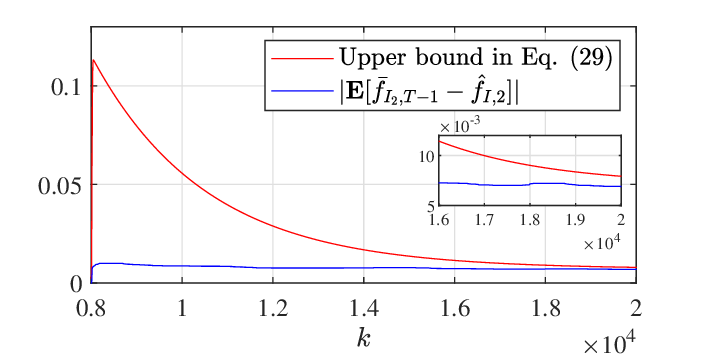}
    \caption{\small Performance bounds derived by Theorem~\ref{Thm} under the incipient power line fault.}
    \label{fig: PFB4SFl}
\end{figure}

\begin{figure}[t]
        \centering
        \includegraphics[width=0.4\linewidth]{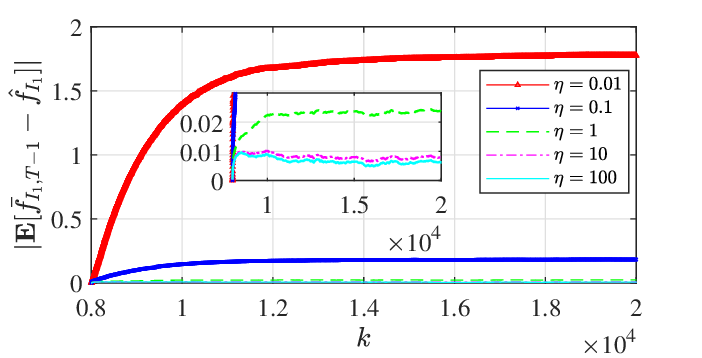}
        \caption{\small Expectation of the fault line current estimation error under different $\eta$.}
        \label{fig: DifferentEta}
\end{figure}

\noindent \textit{\bf Case II: Faults on both power lines $1$ and $2$.}
In the second scenario, we consider a more complex case, where the pole-to-ground short-circuit fault~$f_{L,1}$ happens on the power line between DG~$1$ and DG~$2$, along with an incipient fault~$f_{L,2}$ on the power line between DG~$1$ and DG~$3$. 
The equivalent circuit of the short-circuit fault is depicted in Fig.~\ref{fig: short-circuit fault}, with the dynamics described by:
\begin{align}
    \left\{
    \begin{array}{l}
       \dot{I}_{k}(t)=-\frac{R_{k,1}}{L_{k,1}} {I}_{k}(t)+\frac{1}{L_{k,1}}(V_i(t)-V_f(t))    \\
       \dot{\tilde{I}}_k(t) = -\frac{R_{k,2}}{L_{k,2}} \tilde{I}_{k}(t)+\frac{1}{L_{k,2}}(V_f(t)-V_j(t)) \\
       V_f(t) = (I_k(t) - \tilde{I}_k(t))R_f
    \end{array}
    \right. ,
\end{align}
where $k=1$ and $R_f$ is a low-resistance grounding resistor.  
Other inductance and resistance values depend on the short-circuit fault location. 
The dynamics of $I_k$ can be rewritten in the form of~\eqref{line}, where the power line fault~$f_{L,k}$ becomes
\begin{align*}
    f_{L,k}(t) = \left( \frac{R_k}{L_k} - \frac{R_{k,1}}{L_{k,1}} \right)I_k(t) +  \frac{1}{L_k}V_j(t) - \frac{1}{L_{k,1}}V_f(t) .
\end{align*}
The parameters in the circuit are
$R_{k,1}=R_{k,2}=0.025 ~\si{\ohm}$, $L_{k,1}=L_{k,2}=1~\mu\text{H}$, and $R_f = 0.01~\si{\ohm}$.
The short-circuit fault happens at~$t=100$~ms.
The dynamics of the incipient fault~$f_{L,2}$ follow:~$\dot{f}_{L,2}(t) = -\beta_l f_{L,2}(t) + \beta_l \bar{f}_{L,2}$, with~$\beta_l = 4 \times 10^{-9}$, $\bar{f}_{L,2}=1$, and the occurrence time~$t=50$~ms.
An actuator fault~$f_{a,1}$ is considered in DG unit $1$, which is described by a step signal and changes from $0$ to $0.6$ at~$t=100$~ms.
Additionally, to assess the impact of load variations on faulty current estimation, we introduce a step change in~$P_{2}$, which decreases from $115$~W to $105$~W at $t = 150$ ms.

\begin{figure}[t]
    \centering
    \includegraphics[width=0.4\linewidth]{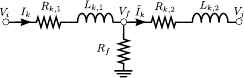}
    \caption{\small Illustration of the short-circuit fault.}
    \label{fig: short-circuit fault}
\end{figure}

% \begin{figure*}[htbp]
%     \centering
%     \subfigure[]{\includegraphics[width=0.32\textwidth]{Figures/Fig_FiEst_SCF_DG1.eps}\label{fig: FlResult_SCF_DG1}}
%     \hfill
%     \subfigure[]
%     {\includegraphics[width=0.32\textwidth]{Figures/Fig_FiEst_SCF_DG2.eps}\label{fig: FlResult_SCF_DG2}}
%     \hfill
%     \subfigure[]
%     {\includegraphics[width=0.32\textwidth]{Figures/Fig_FiEst_SCF_DG3.eps}\label{fig: FlResult_SCF_DG3}}
%     \caption{\small Diagnosis of short-circuit line faults: 
%     (a) diagnosis results of~$\mathcal{D}_{l,1}$,   
%     (b) diagnosis results of~$\mathcal{D}_{l,2}$, and 
%     (c) diagnosis results of~$\mathcal{D}_{l,3}$.}
%     \label{fig: FlResult_SCF}
% \end{figure*}

\begin{figure}[t]
    \centering
%    \captionsetup{justification=centering}
    \includegraphics[width=0.4\linewidth]{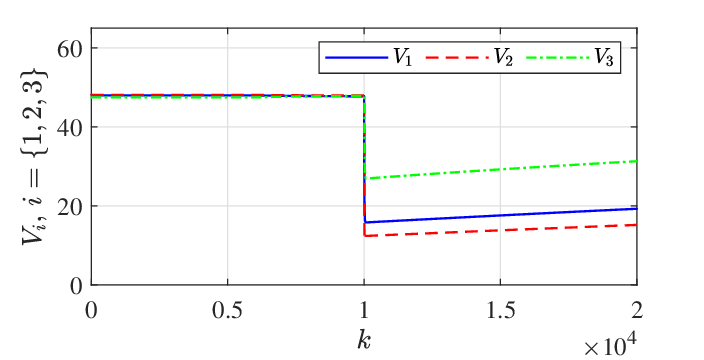}
    \includegraphics[width=0.4\linewidth]{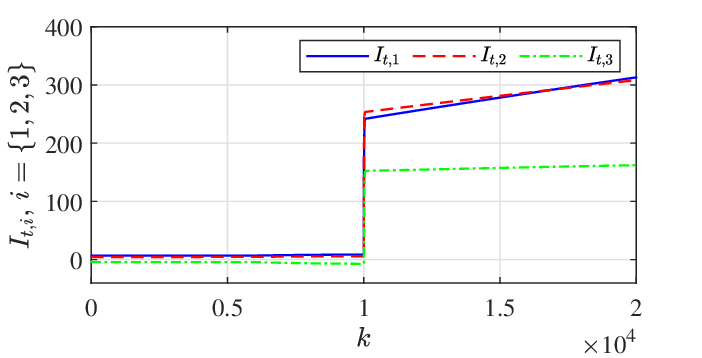}
    \caption{\small Dynamic response of the DC microgrid when considering the incipient and short-circuit line faults.}
    \label{fig: VI4_SCF}
\end{figure}

\begin{figure}[t]
    \centering
    \includegraphics[width=0.95\linewidth]{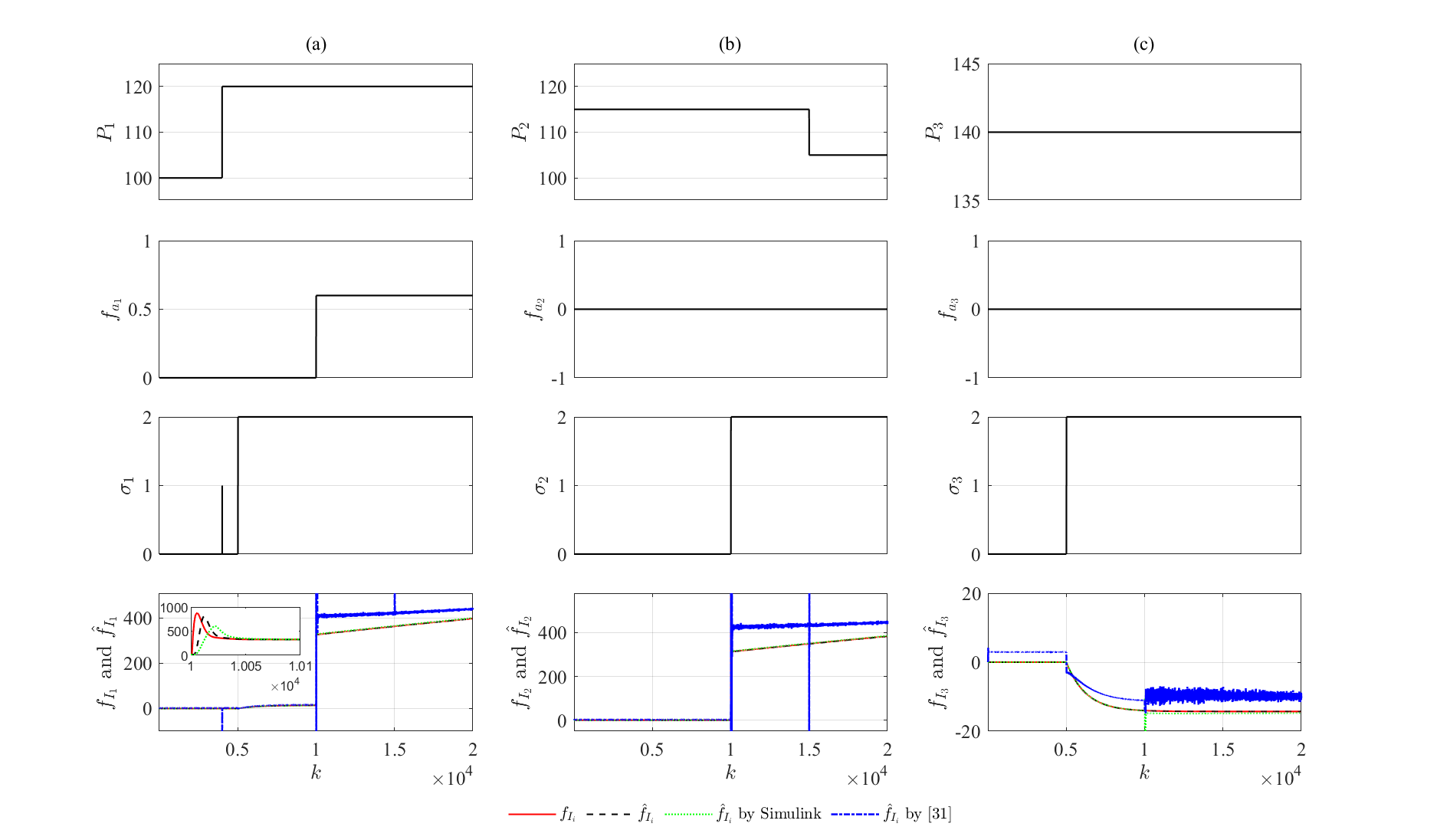}
    \caption{\small Diagnosis of the incipient and short-circuit line faults: (a) results of~$\mathcal{D}_{l,1}$, (b) results of~$\mathcal{D}_{l,2}$, and (c) results of~$\mathcal{D}_{l,3}$.}
    \label{fig: FlResult_SCF}
\end{figure}

\begin{figure*}[htbp]
    \centering
    \includegraphics[width=0.95\textwidth]{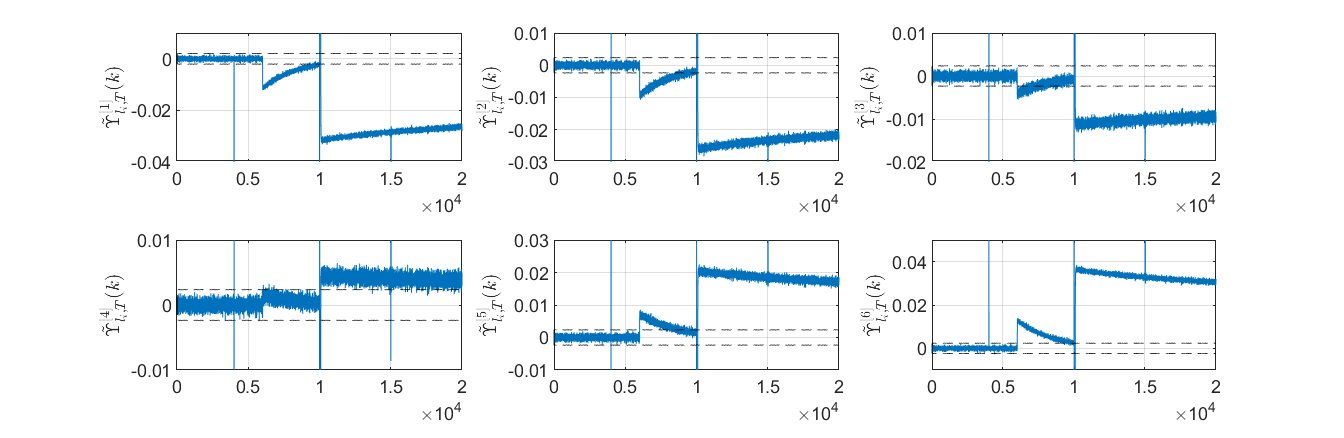}
    \caption{\small Residual $\tilde{\Upsilon}_{1,T}$ for isolation of the incipient and short-circuit line faults.}
    \label{fig: r_tilde4_SCF_DG1}
\end{figure*}

The simulation results of the second scenario are presented in Figs.~\ref{fig: VI4_SCF}-\ref{fig: PFB4_SCF}. 
In particular, Fig.~\ref{fig: VI4_SCF} shows the voltage and current variations of the microgrid under load changes and system faults.
When the incipient power line fault~$f_{L,2}$ happens, its effect on voltages and currents is unobvious from the measurements. 
In contrast, the short-circuit fault~$f_{L,1}$ at~$t=100$~ms introduces a low-resistance~$R_f$, causing sharp voltage and current increases in the microgrid, as shown in Fig.~\ref{fig: VI4_SCF}. 
As a result, the short-circuit fault is easily detectable from measurement signals. 
Nevertheless, the proposed diagnosis approach can provide more detailed information about these faults.

The diagnosis results for the power line faults~$f_{L,1}$ and~$f_{L,2}$ are shown in Fig.~\ref{fig: FlResult_SCF}. Specifically, at $t=40$ ms, the status indicator~$\sigma_1$ briefly switches to~$1$, indicating the detection of the step change in $P_{1}$.
At $t=60$ ms, the incipient fault~$f_{L,2}$ happens on the power line between DG units~$1$ and $3$, then the indicator signals~$\sigma_1$ and~$\sigma_3$ switch to $2$. This means that~$f_{L,2}$ is detected by both $\mathcal{D}_{l,1}$ and $\mathcal{D}_{l,3}$.
At $t=100$ ms, the short-circuit fault~$f_{L,1}$ between DG units $1$ and $2$ occurs, and the indicator signal~$\sigma_2$ becomes $2$, indicating successful detection of~$f_{L,1}$.  
Note that the detection process is further supported by the residual behavior in Fig.~\ref{fig: r_tilde4_SCF_DG1}. Due to the space limitation, we only show the residual $\tilde{\Upsilon}_{1,T}$ here.

The bottom row of~Fig.~\ref{fig: FlResult_SCF} shows estimates of faulty line currents in each DG unit caused by~$f_{L,1}$ and~$f_{L,2}$. 
It can be seen that, faulty currents~$f_{I,1}$ and $f_{I,3}$ caused by the incipient fault~$f_{L,2}$ can be accurately estimated by~$\mathcal{D}_{l,1}$ and $\mathcal{D}_{l,3}$, respectively.
After the short-circuit fault~$f_{L,1}$ happens at~$t=100$ ms, the faulty current~$f_{I,2}$ is estimated by $\mathcal{D}_{l,2}$, while~$\mathcal{D}_{l,1}$ provides a combined estimate of currents from both~$f_{L,1}$ and~$f_{L,2}$.
The comparison with the estimation method in~\cite{van2022multiple} also shows the superiority of the proposed approach.
Finally, Fig.~\ref{fig: PFB4_SCF} validates the developed upper bound on the estimation error in this scenario.

\begin{figure*}[htbp]
    \centering
    \includegraphics[width=0.32\linewidth]{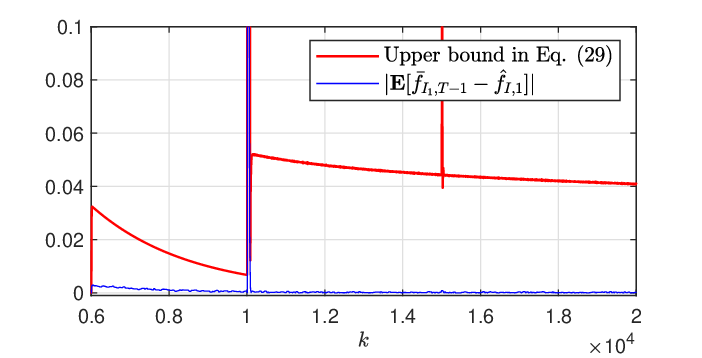}   
    \hspace{-10pt}
    \includegraphics[width=0.32\linewidth]{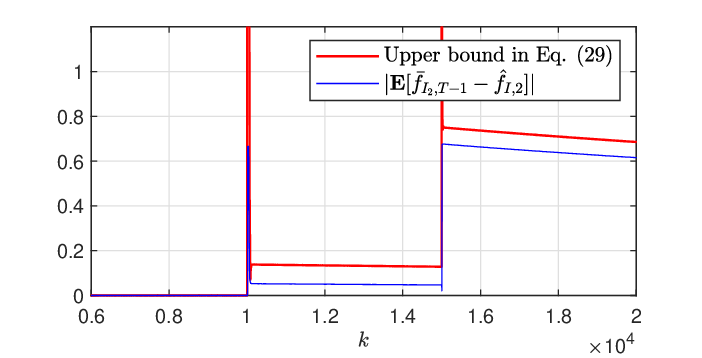}  
    \hspace{-10pt}
    \includegraphics[width=0.32\linewidth]{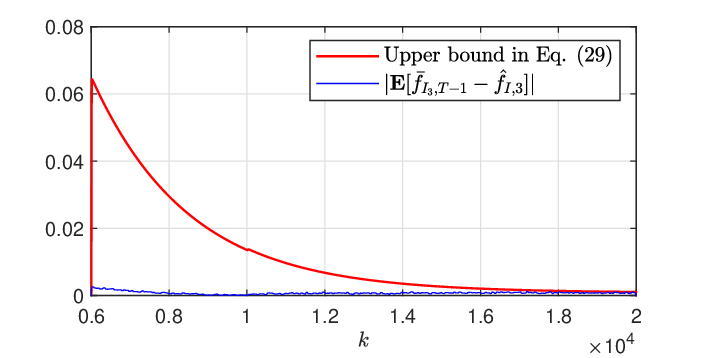}
    \caption{\small Performance bounds derived by Theorem~\ref{Thm} under the incipient and short-circuit power line faults.}
    \label{fig: PFB4_SCF}
\end{figure*}

\subsection{Real-time Feasibility Analysis}
To evaluate the real-time feasibility of the proposed algorithm, we measured the execution time of one sampling step. 
The algorithm was executed for $2\times10^4$ iterations on the above computational platform.
The average execution time per step is $7.8941\times10^{-6}$~s, which is smaller than the sampling time $t_s = 1\times10^{-5}$~s. 
This result indicates that the proposed algorithm can be executed within one sampling period and is therefore suitable for real-time implementation.
It is worth noting that the reported runtime is obtained from a straightforward MATLAB implementation without dedicated code-level optimization. 
Further reduction in execution time is expected when the algorithm is implemented using optimized or compiled code.

%%%%%%%%%%%%%%%%%%%%%%%%%%%%%%%%%%%%%%%%%%%%%%%%%%%%%%%%%%%%%%%%%%%%%%%%%%%%%%%%%%%%%%%%%%%%%%%%%%%%%%%%%%%%
\section{Conclusions}\label{sec: 6}
This paper presents a distributed diagnosis scheme for the detection and estimation of actuator and power line faults in DC microgrids, under the effects of unknown power loads and stochastic noise.
A key contribution is the analysis of the coupling effect between power load variations and faulty line currents, an aspect previously underexplored in DC microgrids.
To address this challenge, we introduce a novel differentiate-before-estimate strategy that enhances fault detection accuracy. 
Future work includes: (i) extending the proposed approach to more general load conditions, such as nonlinear ZIP (constant-impedance, constant-current, and constant-power) loads, where active signal injection may be incorporated to improve the identifiability of a broader class of loads; (ii) accounting for mismatches between local sampling rates and inter-unit communication rates; and (iii) considering additional fault types, such as sensor and parameter faults.

\appendix
{\appendices
\section{}\label{app: A} 
\begin{proof}[Proof of Proposition~\ref{prop: diff_P_fl}]
     In the first part of the proof, we show that the effect of a step load change on $\tilde{\Upsilon}_{i,T}$ will vanish after~$T$ steps. 
    Suppose that $P_{i}$ becomes $\bar{P}_{i} = P_{i} + \Delta P_{i}$ at some time instant $k_0 \in [k-T+1,k-1]$.
    Then, $\BDViT$ can be divided into two parts accordingly, i.e.,
    \begin{align*}
    \BDViT(k-1) = \begin{bmatrix}
        \boldsymbol{\mathcal{V}}_{i,k_0-(k-T+1)}(k_0-1) \\ \boldsymbol{\mathcal{V}}_{i,k-k_0}(k-1)
        \end{bmatrix}.
    \end{align*}
The parity space equation~\eqref{eq: parity-space relation} in the presence of the step load change can be written as
\begin{align*}
\Upsilon_{i,T} = 
       \Psi_{i,T} P_{i} 
    + \begin{bmatrix}
        \mathbf{0} \\ \boldsymbol{\mathcal{V}}_{i,k-k_0}
    \end{bmatrix}\Delta P_{i} +  \Omega_{i,T}.
\end{align*}
Substituting $\Upsilon_{i,T}$ into~\eqref{eq: Pli_est} leads to
\begin{align*}
    \hat{P}_{i} = P_{i} + \Phi_{i,T}\begin{bmatrix}
        \mathbf{0} \\ \boldsymbol{\mathcal{V}}_{i,k-k_0}
    \end{bmatrix} \Delta P_{i} +  \Phi_{i,T} \Omega_{i,T},
\end{align*}
where $\Phi_{i,T}\Psi_{i,T}= {\bf I}$ is used here.
Then, based on~\eqref{eq: r_tilde}, the expectation of the estimation error $\mathbf{E}\left[\tilde{\Upsilon}_{i,T}\right]$ becomes
\begin{align*}
    \mathbf{E}\left[\tilde{\Upsilon}_{i,T}\right] = ({\bf I}_{n_{\Upsilon}} -\Psi_{i,T}  \Phi_{i,T}) \begin{bmatrix}
        \mathbf{0} \\ \boldsymbol{\mathcal{V}}_{i,k-k_0}
    \end{bmatrix} \Delta P_{i},
\end{align*}
which is no longer zero because of $\Delta P_{i}$.
Thus, entries of $\tilde{\Upsilon}_{i,T}$ can exceed the threshold interval when $k\in [k_0+1,k_0+T-2]$. 

For $k \geq k_0+T-1$, namely,~$T$ steps after the power load change, the estimation result in~\eqref{eq: Pli_est} becomes unbiased again because $\Upsilon_{i,T} = \Psi_{i,T} \bar{P}_{i} + \Omega_{i,T}$.
Thus, entries of $\tilde{\Upsilon}_{i,T}$ will lie within the threshold interval with the probability greater than $1-1/\alpha^2$ according to Chebyshev's inequality. This completes the first part of the proof.

In the second part of the proof, to show the effects of line faults on $\tilde{\Upsilon}_{i,T}$, let us rewrite the expression of $\Upsilon_{i,T}$ in~\eqref{eq: parity-space relation} with line faults, which becomes
\begin{align*}
    \Upsilon_{i,T} =\Psi_{i,T}P_{i} + \mathcal{O}^{\bot}_{i,T}  \mathcal{Z}_{i_1,T} \BDfiT + \Omega_{i,T}.
\end{align*}
The residual $\tilde{\Upsilon}_{i,T}$ and its expected value are then given by 
\begin{align*}
    \tilde{\Upsilon}_{i,T}
    &=\Upsilon_{i,T} - \Psi_{i,T}\Phi_{i,T}\Upsilon_{i,T} 
    = ({\bf I}_{n_{\Upsilon}}-\Psi_{i,T} \Phi_{i,T}) ( \mathcal{O}^{\bot}_{i,T}  \mathcal{Z}_{i_1,T}\BDfiT+\Omega_{i,T}), \\
    \mathbf{E}[\tilde{\Upsilon}_{i,T}] &= ({\bf I}_{n_{\Upsilon}}-\Psi_{i,T} \Phi_{i,T})  \mathcal{O}^{\bot}_{i,T}  \mathcal{Z}_{i_1,T} \BDfiT \neq \mathbf{0}.
\end{align*}
Different from the effects of step load changes~$\Delta P_{i}$ that vanish after $T$ steps, the line faults can be distinguished from~$\Delta P_{i}$ if there exists $\kappa \in \{1,\dots,n_{\Upsilon}\}$ such that~$\tilde{\Upsilon}^{[\kappa]}_{i,T}(k) \notin \left[-\varepsilon^{[\kappa]}_{i}(k), \varepsilon^{[\kappa]}_{i}(k)\right]$ for at least $T$ consecutive steps. This completes the proof.
\end{proof}

\section{}\label{app: B}
To prove Theorem~\ref{Thm}, we introduce the following lemma.
\begin{Lem}[Eigenvalue bounds of a matrix]\label{lem: M_eig}
The matrix
\begin{align*}
    \mathcal{M}_{i,T}=
    \begin{bmatrix}
        \mathcal{M}^{[11]}_{i,T} &\mathcal{M}^{[12]}_{i,T}\\
        \mathcal{M}^{[21]}_{i,T} &\mathcal{M}^{[22]}_{i,T}
    \end{bmatrix}
    =\begin{bmatrix}
    \Psi^{\top}_{i,T} \Psi_{i,T}  &\Psi^{\top}_{i,T}\bar{\mathcal{Z}}_{i,T} \\
    \bar{\mathcal{Z}}^{\top}_{i,T} \Psi_{i,T} &\bar{\mathcal{Z}}^{\top}_{i,T} \bar{\mathcal{Z}}_{i,T}
\end{bmatrix}
\end{align*}
is bounded by
$\underline{\lambda}_{\mathcal{M}_{i,T}} \mathbf{I}_2 \preceq \mathcal{M}_{i,T} \preceq \bar{\lambda}_{\mathcal{M}_{i,T}} \mathbf{I}_2$.
\end{Lem}

\begin{proof}
    Since $\mathcal{M}_{i,T}$ is a $2 \times 2$ matrix, the larger eigenvalue of~$\mathcal{M}_{i,T}$ satisfies
\begin{align*} 
    \frac{\mathcal{M}^{[11]}_{i,T} + \mathcal{M}^{[22]}_{i,T} + \sqrt{\Xi_{i,T}}}{2}
    \leq \mathcal{M}^{[11]}_{i,T} + \mathcal{M}^{[22]}_{i,T}
    = \Psi^{\top}_{i,T} \Psi_{i,T} + \bar{\mathcal{Z}}^{\top}_{i,T} \bar{\mathcal{Z}}_{i,T} =\bar{\lambda}_{\mathcal{M}_{i,T}},
\end{align*}
where 
$\Xi_{i,T} = \left(\mathcal{M}^{[11]}_{i,T} + \mathcal{M}^{[22]}_{i,T}\right)^2-4\left(\mathcal{M}^{[11]}_{i,T} \mathcal{M}^{[22]}_{i,T}-{\mathcal{M}^{[12]}_{i,T}}^2\right)$.
The above inequality holds because of the Cauchy-Schwarz inequality, i.e., 
\begin{align*}
    \mathcal{M}^{[11]}_{i,T} \mathcal{M}^{[22]}_{i,T}-{\mathcal{M}^{[12]}_{i,T}}^2 
    = \Psi^{\top}_{i,T} \Psi_{i,T} \bar{\mathcal{Z}}^{\top}_{i,T} \bar{\mathcal{Z}}_{i,T} - \Psi^{\top}_{i,T}\bar{\mathcal{Z}}_{i,T} \bar{\mathcal{Z}}^{\top}_{i,T} \Psi_{i,T} \geq 0.
\end{align*}
For the smaller eigenvalue of~$\mathcal{M}_{i,T}$, it holds that
\begin{align*}
    \frac{\mathcal{M}^{[11]}_{i,T} + \mathcal{M}^{[22]}_{i,T}-\sqrt{\Xi_{i,T}}}{2}  
    =\frac{2\left(\mathcal{M}^{[11]}_{i,T} \mathcal{M}^{[22]}_{i,T}-{\mathcal{M}^{[12]}_{i,T}}^2\right)}{\mathcal{M}^{[11]}_{i,T} + \mathcal{M}^{[22]}_{i,T}+\sqrt{\Xi_{i,T}}} 
    \geq \frac{\mathcal{M}^{[11]}_{i,T} \mathcal{M}^{[22]}_{i,T}-{\mathcal{M}^{[12]}_{i,T}}^2}{\mathcal{M}^{[11]}_{i,T} + \mathcal{M}^{[22]}_{i,T}}
    =\underline{\lambda}_{\mathcal{M}_{i,T}}.
\end{align*}
This is because 
$\Xi_{i,T} \leq \left(\mathcal{M}^{[11]}_{i,T} + \mathcal{M}^{[22]}_{i,T}\right)^2$.
This completes the proof.
\end{proof}

\begin{proof}[Proof of Theorem~\ref{Thm}]
    Based on the analytical solution given in~\eqref{eq: analytical sol}, the estimation error of $\Theta_i$ can be written as
    \begin{align*}
        \Theta_i -  \hat{\Theta}_i 
        = \Theta_i -  \mathcal{K}^{-1}_{i,T}  \left(\Gamma_{i,T}^{\top}\Sigma^{-1}_{\Omega_{i,T}} \Upsilon_{i,T} +\eta \nu_1^{\top}\nu_1 \hat{\Theta}_{i\_} 
        +\eta \nu_1^{\top}\nu_1 \Theta_{i}- \eta \nu_1^{\top}\nu_1 \Theta_{i}\right).
    \end{align*}
    Substituting $\Upsilon_{i,T}$ in~\eqref{eq: Reform_PS} into the above equation leads to
    \begin{align*}
        \Theta_i -  \hat{\Theta}_i 
        =&\Theta_i -\mathcal{K}^{-1}_{i,T}\left( \Gamma_{i,T}^{\top}\Sigma^{-1}_{\Omega_{i,T}} \Gamma_{i,T} + \eta \nu_1^{\top}\nu_1 \right)\Theta_i \\
        &-\mathcal{K}^{-1}_{i,T} \Gamma_{i,T}^{\top}\Sigma^{-1}_{\Omega_{i,T}} \left(\mathcal{O}^{\bot}_{i,T} \mathcal{Z}_{i_1,T} (\Delta \boldsymbol{f}_{I_i,T-1} 
        +\Delta \boldsymbol{P}_{i})  +\Omega_{i,T}\right) 
        +\mathcal{K}^{-1}_{i,T} \eta \nu_1^{\top}\nu_1 (\Theta_{i} - \hat{\Theta}_{i\_}) \\
        = &-\mathcal{K}^{-1}_{i,T} \Gamma_{i,T}^{\top}\Sigma^{-1}_{\Omega_{i,T}} \left(\mathcal{O}^{\bot}_{i,T} \mathcal{Z}_{i_1,T}(\Delta \boldsymbol{f}_{I_i,T-1}+\Delta \boldsymbol{P}_{i})  +\Omega_{i,T}\right) 
        +\mathcal{K}^{-1}_{i,T} \eta \nu_1^{\top}\nu_1 (\Theta_{i} - \hat{\Theta}_{i\_}),
    \end{align*}
    Since~$\bar{f}_{I_i,T-1}-\hat{f}_{I,i} = [0~1](\Theta_i - \hat{\Theta}_i)$ and recall that $\mathcal{K}_{i,T} = \Gamma_{i,T}^{\top}\Sigma^{-1}_{\Omega_{i,T}}\Gamma_{i,T}+\eta \nu_1^{\top}\nu_1$, the absolute value of $\mathbf{E}[\bar{f}_{I_i,T-1}-\hat{f}_{I,i} ]$ is bounded by 
    \begin{align}\label{eq: Error_ineq}
        \left|\mathbf{E}\left[ \bar{f}_{I_i,T-1}-\hat{f}_{I,i} \right] \right| 
        \leq &\left\| \mathcal{K}^{-1}_{i,T}  \right\|_2 \left\|\Gamma_{i,T}^{\top} \right\|_2 \left\|\Sigma^{-1}_{\Omega_{i,T}}\right\|_2 \left\|\mathcal{O}^{\bot}_{i,T} \mathcal{Z}_{i_1,T}\right\|_2\times \\ &\left(\left\|\Delta \boldsymbol{f}_{I_i,T-1} \right\|_2 + \left\|\Delta \boldsymbol{P}_{i} \right\|_2 \right) \notag 
        + 
        |[0 ~1] \mathcal{K}^{-1}_{i,T}  
        \eta \nu_1^{\top}\nu_1 (\Theta_i - \hat{\Theta}_{i\_}) |. 
    \end{align}
    According to Lemma \ref{lem: M_eig}, matrix $\mathcal{K}_{i,T}$ is lower bounded by
    \begin{align*}
     \mathcal{K}_{i,T}  &=   \Gamma_{i,T}^{\top}\Sigma^{-1}_{\Omega_{i,T}}\Gamma_{i,T}+\eta \nu_1^{\top}\nu_1 \\ 
    &\succeq 
     \frac{1}{\bar{\lambda}_{\Sigma_{\Omega_{i,T}}}}  
    \begin{bmatrix}
        \Psi^{\top}_{i,T} \Psi_{i,T}  &\Psi^{\top}_{i,T}  \bar{\mathcal{Z}}_{i,T} \\ \bar{\mathcal{Z}}^{\top}_{i,T} \Psi_{i,T} &\bar{\mathcal{Z}}^{\top}_{i,T}\bar{\mathcal{Z}}_{i,T}
    \end{bmatrix} 
     +\eta\nu_1^{\top}\nu_1  \\
    &\succeq   \frac{\underline{\lambda}_{\mathcal{M}_{i,T}}}{\bar{\lambda}_{\Sigma_{\Omega_{i,T}}}}  I +\eta\nu_1^{\top}\nu_1.
    \end{align*}
    Moreover, since the $2$-norm $\| \mathcal{K}^{-1}_{i,T} \|_2  = \bar{\lambda}_{\mathcal{K}^{-1}_{i,T}} = 1/\underline{\lambda}_{\mathcal{K}_{i,T}}$, $\| \mathcal{K}^{-1}_{i,T} \|_2$ is upper bounded by 
    \begin{align}\label{eq: KiT_ineq}
        \left\| \mathcal{K}^{-1}_{i,T} \right\|_2 \leq \frac{\bar{\lambda}_{\Sigma_{\Omega_{i,T}}}}{\underline{\lambda}_{\mathcal{M}_{i,T}}+\eta \bar{\lambda}_{\Sigma_{\Omega_{i,T}}}}.
    \end{align}
    For $\left\|\Gamma_{i,T}^{\top}\right\|_2$, according to Lemma~\ref{lem: M_eig}, we have
    \begin{align}\label{eq: GammaiT_ineq}
        \left\|\Gamma_{i,T}^{\top}\right\|_2 = \sqrt{\left\| \Gamma_{i,T}^{\top} \Gamma_{i,T} \right\|_2} = \sqrt{ \bar{\lambda}_{\mathcal{M}_{i,T}}}.
    \end{align}
    Together with~\eqref{eq: KiT_ineq} and~\eqref{eq: GammaiT_ineq}, the first component on the right-hand side of inequality~\eqref{eq: Error_ineq} is upper bounded by
    \begin{align}\label{eq: ineq_part1}
       \left\| \mathcal{K}^{-1}_{i,T}  \right\|_2 \left\|\Gamma_{i,T}^{\top} \right\|_2 \left\|\Sigma^{-1}_{\Omega_{i,T}}\right\|_2 \left\|\mathcal{O}^{\bot}_{i,T} \mathcal{Z}_{i_1,T}\right\|_2  
        \leq \frac{\bar{\lambda}_{\Sigma_{\Omega_{i,T}}}}{\underline{\lambda}_{\mathcal{M}_{i,T}}+\eta \bar{\lambda}_{\Sigma_{\Omega_{i,T}}}} 
        \frac{\sqrt{ \bar{\lambda}_{\mathcal{M}_{i,T}}}} {\underline{\lambda}_{\Sigma_{\Omega_{i,T}}}} \bar{\sigma}(\mathcal{O}^{\bot}_{i,T} \mathcal{Z}_{i_1,T}).
    \end{align}

  For the second term on the right-hand side of inequality~\eqref{eq: Error_ineq}, it holds that
    \begin{align*}
        [0 ~1] \mathcal{K}^{-1}_{i,T}  
        \eta \nu_1^{\top}\nu_1 (\Theta_i - \hat{\Theta}_{i\_})
        =\frac{ -\bar{\mathcal{Z}}^{\top}_{i,T} \Sigma^{-1}_{\Omega_{i,T}} \Psi_{i,T} \eta (P_{i} - \hat{P}_{i\_})}{\text{det}(\mathcal{K}_{i,T} )}, 
    \end{align*}
    where the inverse term~$\mathcal{K}^{-1}_{i,T}$ is computed in the above equation. 
    The determinant of $\mathcal{K}_{i,T}$ satisfies the following inequality:
    \begin{align*}
        \text{det}(\mathcal{K}_{i,T} ) 
        = &\left(\bar{\mathcal{Z}}^{\top}_{i,T} \Sigma^{-1}_{\Omega_{i,T}} \bar{\mathcal{Z}}_{i,T}\right)
        \left(\Psi^{\top}_{i,T} \Sigma^{-1}_{\Omega_{i,T}} \Psi_{i,T} + \eta\right) 
        - \left(\Psi^{\top}_{i,T} \Sigma^{-1}_{\Omega_{i,T}} \bar{\mathcal{Z}}_{i,T}\right)
        \left(\bar{\mathcal{Z}}^{\top}_{i,T} \Sigma^{-1}_{\Omega_{i,T}} \Psi_{i,T}\right) \\
        &\geq \bar{\mathcal{Z}}^{\top}_{i,T} \Sigma^{-1}_{\Omega_{i,T}} \bar{\mathcal{Z}}_{i,T}\eta.
    \end{align*}
    As a result, we have
    \begin{align}\label{eq: ineq_part2}
        \left|[0 ~1] \mathcal{K}^{-1}_{i,T}  
        \eta \nu_1^{\top}\nu_1 (\Theta_i - \hat{\Theta}_{i\_}) \right| 
        \leq \frac{\bar{\mathcal{Z}}^{\top}_{i,T} \Sigma^{-1}_{\Omega_{i,T}} \Psi_{i,T} |P_{i} - \hat{P}_{i\_}|}{ \bar{\mathcal{Z}}^{\top}_{i,T} \Sigma^{-1}_{\Omega_{i,T}} \bar{\mathcal{Z}}^{\top}_{i,T}} 
        \leq {\frac{\bar{\lambda}_{\Sigma_{\Omega_{i,T}}} \bar{\mathcal{Z}}^{\top}_{i,T}  \Psi_{i,T} |P_{i} - \hat{P}_{i\_}|}{\underline{\lambda}_{\Sigma_{\Omega_{i,T}}}\bar{\mathcal{Z}}^{\top}_{i,T}  \bar{\mathcal{Z}}_{i,T}} }.
    \end{align}
    Together with inequalities~\eqref{eq: ineq_part1} and~\eqref{eq: ineq_part2}, the estimation error bound \eqref{eq: error bound} is derived. This completes the proof.
\end{proof}

\section{}\label{app: C}
In this part, we analyze the computational complexity of the proposed diagnosis algorithm for fault line current estimation to demonstrate its real-time feasibility. Specifically, we provide a step-by-step analysis of the main computational procedures executed at each sampling instant for power line fault diagnosis. 
Let us first recall the definitions of some parameters.
\begin{itemize}
    \item $d_{\mathcal{N}}+1$: the order of the residual generators in~\eqref{eq: filter},~\eqref{eq: filter fl}, and~\eqref{eq: prefilter 2}.
    \item $n_y$: the dimension of the measurement signal.
    \item $n$: the number of buses involved in the line-current estimation step in~\eqref{eq: Ik_est}.
    \item $m$: the number of buses involved in the line-current estimation step in~\eqref{eq: Ik_est}.
    \item $n_{\Upsilon}$: the dimension of the projected stacked residual $\Upsilon_{i,T}$ in~\eqref{eq: parity-space relation}.
    \item $T$: the length of the sliding window.
\end{itemize}
    We also distinguish offline precomputable quantities from online computations. 
    In particular, matrices such as $\mathcal{O}^{\bot}_{i,T} \mathcal{Z}_{i_1,T}$, $\mathcal{O}^{\bot}_{i,T}\mathcal{Z}_{i_2,T}$, and ${\Sigma^{-1}_{\Omega_{i,T}}}$ are treated as offline quantities and therefore do not contribute to the online complexity.

    \begin{itemize}
        \item[(1)] \textbf{Residual~$r_i$ in~\eqref{eq: filter fl}.} 
        The residual generator is implemented as a filter of order $d_{\mathcal{N}}+1$ with input dimension $n_y+1$. The online complexity is dominated by the filter state update and output computation. Therefore, the computational complexity of generating $r_i$ is $O(d^2_{\mathcal{N}} + d_{\mathcal{N}} n_y)$.

        \item[(2)] \textbf{Computation of the estimated line current $\hat{I}_{k,h}$ in~\eqref{eq: Ik_est}.} 
        This step involves a scalar state update and a weighted summation of the bus voltages. Since the dominant operation is the summation over the $n$ bus-voltage terms, its complexity scales linearly with the network dimension. Hence, the complexity of this step is $O(n)$.

        \item[(3)] \textbf{Generation of $\hat{r}_{I,i}$ in~\eqref{eq: prefilter 2}.}
        This step first forms a weighted summation of $m$ scalar estimated line currents and then propagates a scalar-input scalar-output filter of order $d_{\mathcal{N}}+1$.
        Accordingly, the online complexity consists of a linear term in $m$ and a filter-related term of order $d_{\mathcal{N}}$. The resulting complexity is $O(d^2_{\mathcal{N}} + m)$.

        \item[(4)] \textbf{Computation of the residual $\tilde{r}_{i}$ in~\eqref{eq: til_ri}.}
        Since both $r_i$ and $\hat{r}_{I,i}$ are scalar, this operation only requires one subtraction. Thus, its complexity is constant $O(1)$.

        \item[(5)] \textbf{Construction of the projected stacked residual $\Upsilon_{i,T}$ and the matrix~$\Psi_{i,T}$ in~\eqref{eq: parity-space relation}.}
        The projected residual used in the estimation stage is not the raw stacked residual, but its projection after eliminating the effect of the initial condition. Its dimension is~$n_{\Upsilon}$.
        The online construction of the regressor~$\Psi_{i,T}$ involves multiplying a precomputed matrix by a windowed input vector of length proportional to $T$.
        Therefore, the complexity of this step is $O(n_{\Upsilon} T)$.  

        \item[(6)] \textbf{Computation of the matrix~$\Phi_{i,T}$ and the estimation of $P_i$ in~\eqref{eq: Pli_est}.} 
        Once $\Psi_{i,T}$ is available, the weighted LS estimate of $P_i$ involves matrix-vector products with ${\Sigma^{-1}_{\Omega_{i,T}}}$, inner products, and a final scalar update.
        Since $\Psi_{i,T} \in \mathbb{R}^{n_{\Upsilon}}$, the dominant online computation scales quadratically with $n_{\Upsilon}$. The multiplication used to obtain $\hat{P}_i$ from $\Phi_{i,T} \Upsilon_{i,T}$ only adds a liner term~$O(n_{\Upsilon})$, which does not change the overall order. Therefore, the complexity of this step is $O(n^2_{\Upsilon}+n_{\Upsilon})$.

        \item[(7)] \textbf{Construction of the compensated projected residual $\tilde{\Upsilon}_{i,T}$ in~\eqref{eq: r_tilde}.}
        This step is a simple vector subtraction of dimension~$n_{\Upsilon}$. The complexity is $O(n_{\Upsilon})$.

        \item[(8)] \textbf{Threshold evaluation for fault and load change detection.} 
        The decision logic checks each component of $\tilde{\Upsilon}_{i,T}$ against its corresponding threshold. Since this is an element-wise comparison over $n_{\Upsilon}$ components, the complexity is $O(n_{\Upsilon})$.

        \item[(9)] \textbf{Fault-current estimation after fault detection in~\eqref{eq: analytical sol}.}
        After a fault is detected, the fault-current estimate is obtained from the analytical solution of a quadratic optimization problem with a two-dimensional parameter vector. The online construction of the corresponding regressor matrix $\Gamma_{i,T}$ has complexity $O(n_{\Upsilon}T)$, while the computation of the associated weighted LS terms scales as $O(n^2_{\Upsilon})$.
        
        Since the optimization variable is only two-dimensional, the inversion of the resulting $2 \times 2$ matrix and the extraction of the final fault-current estimate have constant complexity. Therefore, the overall online complexity of the fault-current estimation step is $O(n^2_{\Upsilon} + n_{\Upsilon}T)$.
    \end{itemize}
    Based on the above analysis, the overall online computational complexity of the main diagnosis procedure per sampling step can be summarized as
    \begin{align*}
        O(d^2_{\mathcal{N}} + d_{\mathcal{N}} n_y + n^2_{\Upsilon} + n_{\Upsilon}T+ m+n).
    \end{align*}
    Moreover, the fault-current estimation step is only activated after a fault is detected, and therefore, it is event-triggered rather than continuously executed at every sampling instant. 
    This further reduces the average online computational burden in practice.
    Note that the obtained complexity is polynomial in the filter order, the input dimension, and the projected residual dimension, and does not involve iterative online optimization. Therefore, the proposed method has a relatively low online computational burden and is suitable for real-time applications.

\section{}\label{app: D}
In this section, we present the results of the proposed power line fault diagnosis method under non-step power load variations. 

\textbf{Case~1: non-ideal step loads (e.g., EV-charger start-up).}
    To illustrate this behavior, we consider an EV-charger-like start-up transient modeled as a first-order exponential rise:
    \begin{align*}
        P_i(t)=P_i(t_0)+\Delta P_i\big(1-e^{-(t-t_0)/\tau}\big), \quad t\ge t_0,
    \end{align*}
    where $t_0$ is the onset time and $\tau$ is the time constant. 
    Simulation results with $\tau=0.01$ (slow rise) and $\tau=0.001$ (fast rise) are reported in Fig.~\ref{fig: varying load}. 
    As shown in the figure, slow transients may not trigger the load-change indicator, whereas fast transients are more likely to be detected since they resemble step-like changes. 
    In both cases, the proposed line-fault estimation performance is maintained.
     \begin{figure*}[h]
        \centering
        \includegraphics[width=0.3\linewidth]{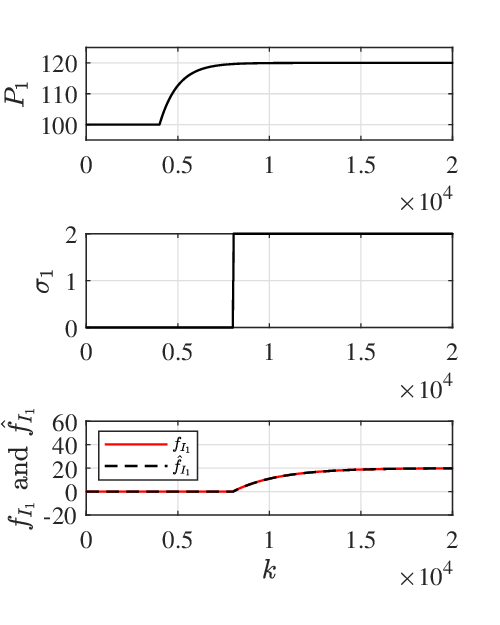}   
        \hspace{20pt}
        \includegraphics[width=0.3\linewidth]{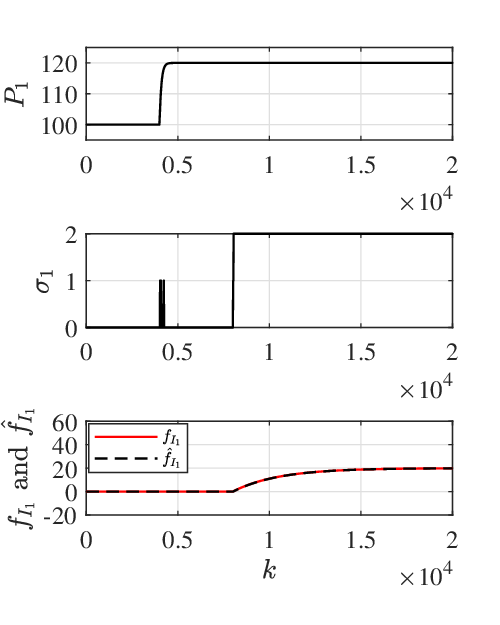}  
        \caption{\small Power line fault diagnosis results under non-ideal step loads.}
        \label{fig: varying load}
    \end{figure*} 

    \textbf{Case~2: fast-varying load variations.}
    If the load variation is sufficiently fast and large such that the residual exceeds the threshold and remains above it for longer than $T$ steps, the differentiator may misclassify the load variation as a power line fault, which can affect the fault estimation accuracy.
    To demonstrate this conservative scenario, we consider a sinusoidal load fluctuation:
    \begin{align*}
        P_i(t)=P_i(t_0)+\Delta P_i\sin\big(\omega_p(t-t_0)\big),\quad t\ge t_0.
    \end{align*}
    The simulation results in Fig.~\ref{fig: large varying load} show that the indicator function may switch to the ``fault'' mode even in the absence of actual line faults, leading to a nonzero estimated faulty current.
    Nevertheless, due to the regularization term adopted in our estimation scheme, the resulting spurious estimate remains bounded and does not diverge, and the estimator can still track the actual faulty current once a real line fault occurs.
    For comparison, we also applied the method in~\cite{van2022multiple} under the same fast time-varying load condition.
    The results indicate that the baseline method fails to provide reliable line-fault current estimation in this scenario, whereas the proposed scheme maintains a certain level of estimation capability.
    \begin{figure*}[h]
        \centering
        \includegraphics[width=0.6\linewidth]{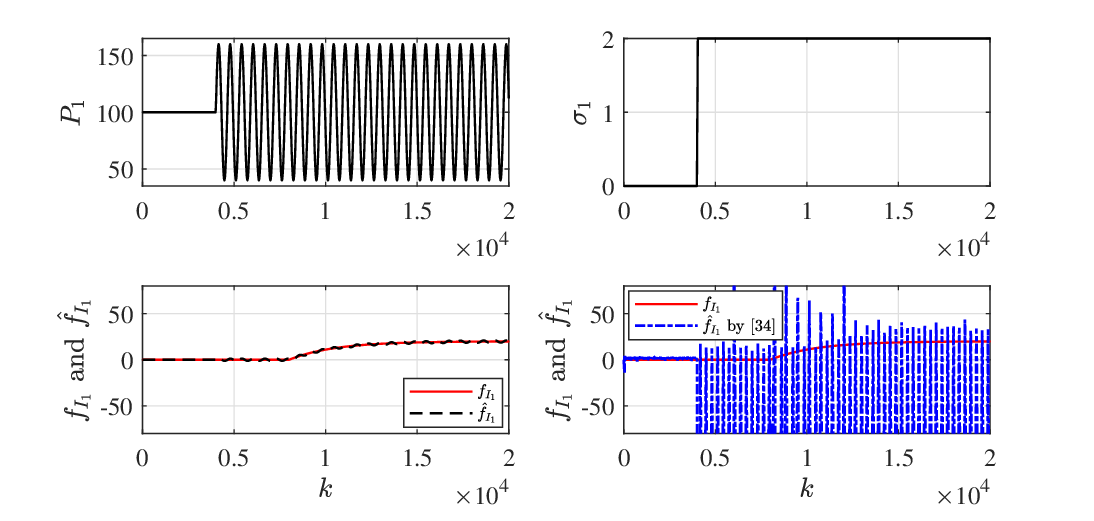}   
        \caption{\small Power line fault diagnosis results under fast-varying loads.}
        \label{fig: large varying load}
    \end{figure*}  

    Overall, the impact of non-step-like load dynamics depends on their magnitude and time scale relative to the window $T$. 
    Particularly, mild and smooth variations may remain below the threshold and have little effect on load or line-fault current estimation, whereas fast and large variations may keep the residual above the threshold for more than $T$ steps, leading to misclassification as line faults and degraded diagnosis performance. A possible remedy is to adopt a bounded-variation load model, e.g., $\|\Delta P_i\|\le \bar{\Delta}_{P}$ within each window, and incorporate this bound into the threshold design.

\section{}\label{app: E}
In this section, we further investigate the influence of additional fault types, namely sensor faults and parameter (multiplicative) faults, on the proposed diagnosis framework. 

\textbf{Case~1: Diagnosis results of sensor faults.}
First, we consider sensor faults in the voltage measurement~$V_i$ and the filter current measurement $I_{t,i}$, denoted by $f_{V_i}$ and $f_{I_{t,i}}$, respectively. Since the control law is state-feedback based, sensor faults affect only the measurement equation.
Therefore, the corresponding state-space model in the presence of sensor faults can be written as
\begin{equation*}
\left\{
    \begin{array}{l}
    \dot{x}_i(t)=A_{i} x_i(t)+B_{i} V^*_i+ D_i d_i(t) + E_i f_{a,i}(t)  +\delta_i(t),\\
    y_i(t)=C_i x_i(t) + C_{s,i} f_{s,i}(t) +  \zeta_{i}(t),
    \end{array}
\right.
\end{equation*}
where $f_{s,i}(t) = [f_{V_i}(t) ~f_{I_{t,i}}(t)]^{\top}$ and $C_{s,i} = \begin{bmatrix}
            {\bf I}_2 \\ {\bf 0}_{1 \times 2}
\end{bmatrix}$.    
After incorporating sensor faults into the DAE framework, the equation~\eqref{eq: DAE fa} used for actuator fault estimation becomes
\begin{align*}
    H_i(p) X_i +  \mathcal{B}_i Y_i + \mathcal{E}_i f_{a,i} + \mathcal{E}_{s,i} f_{s,i} + \omega_i = 0,
\end{align*}
where $\mathcal{E}_{s,i} = \begin{bmatrix} {\bf 0} \\ C_{s,i} \end{bmatrix}$. Accordingly, the actuator fault estimator output~\eqref{eq: filter} becomes
\begin{align*}
    \hat{f}_{a,i} = \frac{N_i(p)\mathcal{E}_i}{a(p)}f_{a,i} + \frac{N_i(p)\mathcal{E}_{s,i}}{a(p)}f_{s,i} + \frac{N_i(p)}{a(p)}\omega_i.
\end{align*}
This expression shows that sensor faults enter the estimator through an additional channel and may affect the estimation output. Similarly, in the presence of sensor faults, the residual~\eqref{eq: filter fl} used for line fault detection becomes
\begin{align*}
    r_{i} = \frac{\mathcal{N}_i(p)\mathcal{G}_i}{a(p)}d_i +\frac{\mathcal{N}_i(p)\mathcal{E}_{s,i}}{a(p)}f_{s,i} + \frac{\mathcal{N}_i(p)}{a(p)} \omega_i.
\end{align*}
Therefore, sensor faults can, in principle, be detected by both the actuator fault estimator and line fault estimator.

    \begin{figure*}[t]
        \centering
        \captionsetup{skip=4pt}
        \includegraphics[width=0.3\linewidth]{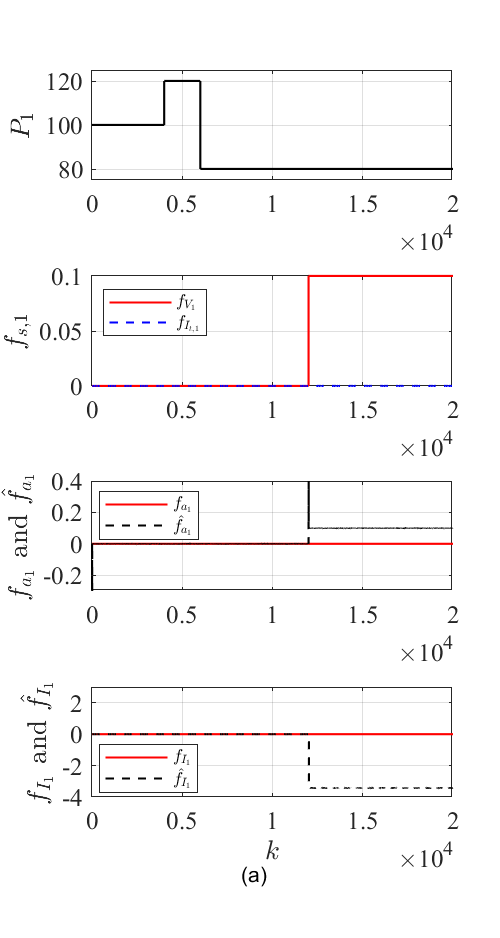} 
        \hspace{30pt}
        \includegraphics[width=0.3\linewidth]{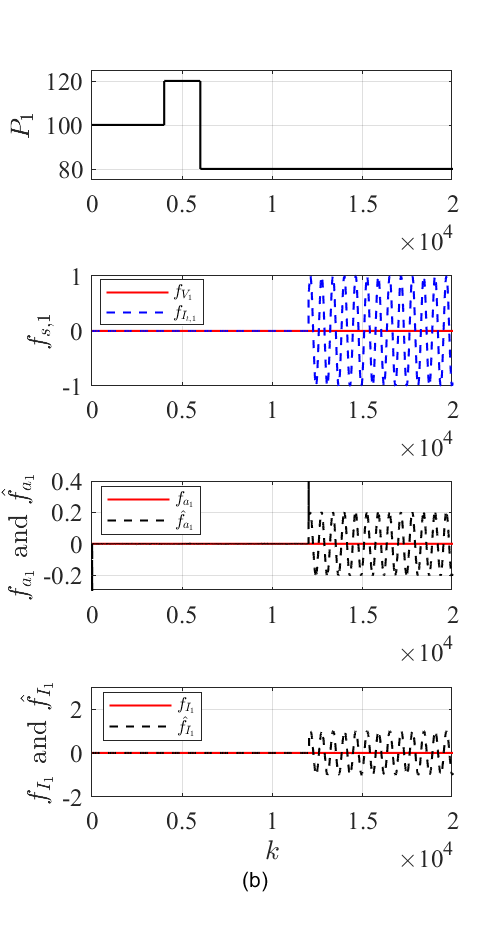} 
        \caption{\small Effects of sensor faults on the diagnosis framework.}
        \label{fig: SensorFault}
    \end{figure*}     

    The corresponding simulation results are presented in Fig.~\ref{fig: SensorFault}, where we consider two sensor fault scenarios: (i) voltage sensor fault~$[f_{V_i} ~f_{I_{t,1}}]^{\top}= [0.1 ~0]^{\top}$ and (ii) current sensor fault $[f_{V_i} ~f_{I_{t,1}}]^{\top}= [0 ~ \sin(1000t)]^{\top}$.
    One can see from Fig.~\ref{fig: SensorFault} that in both cases, the faults are detected. 
    Furthermore, simulation results indicate that voltage sensor faults generally have a stronger impact than current sensor faults. This is because inaccurate voltage measurements directly affect the estimation of line currents, making the line fault detector more sensitive to such faults.

    We next consider parameter faults, arising from variations in the RLC filter parameters, including $R_{t,i}$, $L_{t,i}$, and $C_{t,i}$.
    Such faults alter the system matrices and are therefore commonly referred to as multiplicative faults. 
    In the presence of parameter faults, the system model can be expressed as
    \begin{equation*}
    \left\{
        \begin{array}{l}
        \dot{x}_i(t)=A_{i} x_i(t)+B_{i} V^*_i+ D_i d_i(t) + E_i f_{a,i}(t) + f_{p,i}(t)  +\delta_i(t),\\
        y_i(t)=C_i x_i(t) +  \zeta_{i}(t),
    	\end{array}
    \right.
    \end{equation*}
    where the multiplicative fault is equivalently represented as an additive fault term~$f_{p,i}$ with the following form
    \begin{align*}
        f_{p,i}(t) = 
        \begin{bmatrix}
            A_{f,i} &D_{f,i} &E_{f,i}
        \end{bmatrix}
        \begin{bmatrix}
            x_i(t) \\ d_i(t)\\ f_{a,i}(t)
        \end{bmatrix}.
    \end{align*}
    The matrices $A_{f,i}$, $D_{f,i}$, and $E_{f,i}$ represent unknown fault matrices caused by the changes in system parameters.

    Similar to the above analysis, in the presence of parameter faults, the residuals of the actuator fault estimator and the power line fault estimator become
    \begin{align*}
        \hat{f}_{a,i} &= \frac{N_i(p)\mathcal{E}_i}{a(p)}f_{a,i} + \frac{N_i(p)}{a(p)} \begin{bmatrix}
            I \\0
        \end{bmatrix} f_{p,i} + \frac{N_i(p)}{a(p)}\omega_i, \\
        r_{i} & = \frac{\mathcal{N}_i(p)\mathcal{G}_i}{a(p)}d_i +\frac{\mathcal{N}_i(p)}{a(p)}\begin{bmatrix}
            I \\0
        \end{bmatrix} f_{p,i} + \frac{\mathcal{N}_i(p)}{a(p)} \omega_i.
    \end{align*}
    This indicates that parameter faults also enter the estimators and residual generators through additional channels and can, in principle, be detected within the current diagnosis framework.
    
    The corresponding simulation results are shown in Fig.~\ref{fig: ParamFault}, where we consider two parameter fault scenarios: (i) $L_{t,i} \rightarrow 0.2 L_{t,i}$, and (ii) $C_{t,i} \rightarrow 0.6 C_{t,i}$.
    The results show that the inductor fault is detected by the actuator fault detector~$\mathcal{D}_{a,i}$ but not by the line fault detector~$\mathcal{D}_{l,i}$, whereas the capacitor fault mainly affects~$\mathcal{D}_{l,i}$.
    This difference is due to the distinct roles of~$L_{t,i}$  and ~$C_{t,i}$ in the system dynamics: $L_{t,i}$ is more directly involved in the voltage dynamics related to actuator fault estimation, while $C_{t,i}$ has a stronger influence on the current dynamics used for line fault estimation. 
    \begin{figure*}[t]
        \centering
        \captionsetup{skip=4pt}
        \includegraphics[width=0.3\linewidth]{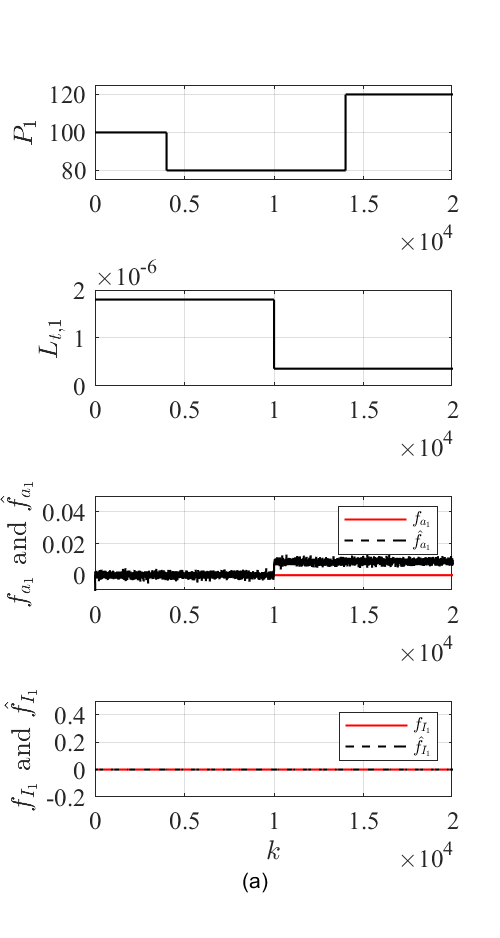} 
        \hspace{30pt}
        \includegraphics[width=0.3\linewidth]{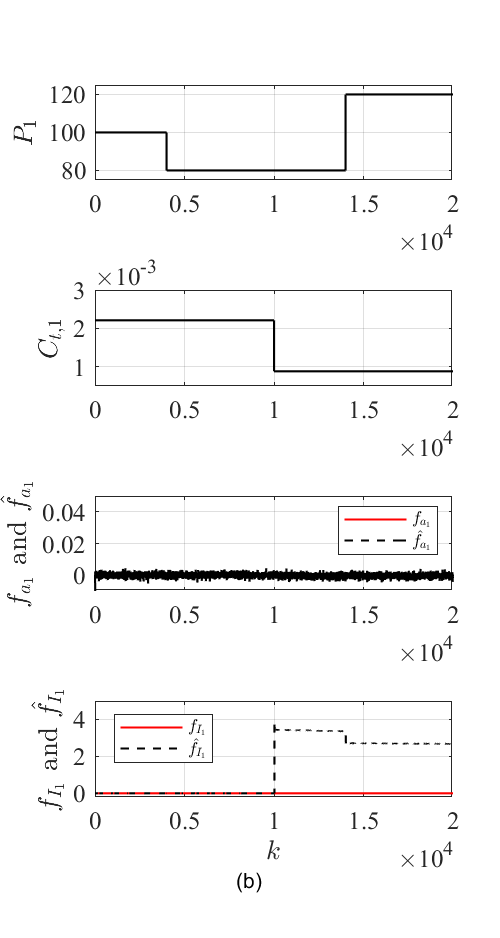} 
        \caption{\small Effects of parameter faults on the diagnosis framework.}
        \label{fig: ParamFault}
    \end{figure*} 
}

% \appendix
% \section{Dynamics of the voltage and current controllers}

\bibliographystyle{elsarticle-num}
\bibliography{bibliography_abrv.bib}
\end{document}